\documentclass{lmcs}
\usepackage[utf8]{inputenc}
\pdfoutput=1

\usepackage{lastpage}
\lmcsdoi{19}{1}{21}
\lmcsheading{}{\pageref{LastPage}}{}{}%
{Nov.~03,~2021}{Mar.~24,~2023}{}

\usepackage{hyperref}
\usepackage{xspace}
\usepackage{stmaryrd}
\usepackage{mathpartir}
\usepackage{amssymb}
\usepackage{tikz}
\usetikzlibrary{positioning}
\usepackage{listings}

\renewcommand{\l}{\ensuremath{\lambda}\xspace}
\newcommand{\lterms}{\ensuremath{\Lambda}\xspace}
\newcommand{\lc}{\ensuremath{\l_{\textsf{c}}}\xspace}
\newcommand{\lcterms}{\ensuremath{\lterms_{\textsf{c}}}\xspace}
\newcommand{\lscbn}{\ensuremath{\l_{\textsf{sn}}}\xspace}
\newcommand{\lwcbn}{\ensuremath{\l_{\textsf{wn}}}\xspace}
\newcommand{\lscbnp}{\ensuremath{\l_{\textsf{sn+}}}\xspace}
\newcommand{\app}{\;}
\newcommand{\letin}{\textsf{let}\xspace}
\newcommand{\esub}[2]{[#1\backslash #2]}
\newcommand{\msub}[2]{\{#1\backslash #2\}}
\newcommand{\re}{\ensuremath{\to}\xspace}
\newcommand{\reb}{\ensuremath{\re_{\beta}}\xspace}
\newcommand{\rec}{\ensuremath{\re_{\textsf{c}}}\xspace}
\newcommand{\rscbn}{\textsf{sn}\xspace}
\newcommand{\rscbnp}{\textsf{sn+}\xspace}
\newcommand{\rex}[1]{\ensuremath{\xrightarrow{#1}}\xspace}
\newcommand{\rescbn}[3]{\ensuremath{\xrightarrow{#1,#2,#3}_{\rscbn}}\xspace}
\newcommand{\rescbnp}[3]{\ensuremath{\xrightarrow{#1,#2,#3}_{\rscbnp}}\xspace}
\newcommand{\literescbn}[1]{\ensuremath{\xrightarrow{\textsf{#1}}_{\rscbn}}\xspace}
\newcommand{\literescbnp}[1]{\ensuremath{\xrightarrow{\textsf{#1}}_{\rscbnp}}\xspace}
\newcommand{\red}{\ensuremath{\to_{\rscbn}}\xspace}

\newcommand{\rdb}{\textsf{dB}\xspace}
\newcommand{\redb}{\ensuremath{\re_{\rdb}}\xspace}
\newcommand{\rlsv}{\textsf{lsv}\xspace}
\newcommand{\relsv}{\ensuremath{\re_{\rlsv}}\xspace}
\newcommand{\rs}[2]{\ensuremath{\textsf{sub}_{#1\smallsetminus #2}}\xspace}
\newcommand{\ri}[1]{\ensuremath{\textsf{id}_{#1}}\xspace}
\newcommand{\auxdb}{\ensuremath{\re_{\rdb}}\xspace}
\newcommand{\wauxlsv}{\ensuremath{\re_{\rlsv}}\xspace}
\newcommand{\auxlsv}[2]{\ensuremath{\xrightarrow{#1, #2}_{\rlsv}}\xspace}
\newcommand{\rgc}{\textsf{gc}\xspace}
\newcommand{\regc}{\ensuremath{\re_{\rgc}}\xspace}
\newcommand{\cx}{\ensuremath{\mathcal{C}}\xspace}
\newcommand{\ctx}[1][\bullet]{\ensuremath{\cx(#1)}\xspace}
\newcommand{\ecx}{\ensuremath{\mathcal{E}}\xspace}
\newcommand{\ectx}[1][\bullet]{\ensuremath{\ecx(#1)}\xspace}
\newcommand{\lctx}{\ensuremath{\mathcal{L}}\xspace}
\newcommand{\plug}[1]{\!|#1|\!}
\newcommand{\fplug}[1]{(\plug{#1})}
\newcommand{\fv}{\textsf{fv}}

\newcommand{\struct}[1][\varphi]{\ensuremath{\mathcal{S}_{#1}}\xspace}
\newcommand{\nf}[1][\varphi]{\ensuremath{\mathcal{N}_{#1}}\xspace}
\newcommand{\set}[1]{\ensuremath{\{#1\}}\xspace}
\newcommand{\mset}[1]{\ensuremath{\{\!\!\{#1\}\!\!\}}\xspace}
\newcommand{\mvar}{\ensuremath{\mathcal{M}}\xspace}
\newcommand{\hole}{\ensuremath{\bullet}\xspace}
\newcommand{\unfold}{^\star}
\newcommand{\revert}{^\dagger}
\newcommand{\arr}{\ensuremath{\rightarrow}\xspace}
\newcommand{\jdg}[3]{\ensuremath{#1\vdash #2:#3}\xspace}
\newcommand{\deriv}[4][\Phi]{\ensuremath{#1\vartriangleright\jdg{#2}{#3}{#4}}\xspace}
\newcommand{\ajdg}[5]{\ensuremath{#1\vdash_{#2}^{#3} #4:#5}\xspace}
\newcommand{\aderiv}[6][\Phi]{\ensuremath{#1\vartriangleright\ajdg{#2}{#3}{#4}{#5}{#6}}\xspace}
\newcommand{\tpos}{\ensuremath{\mathcal{T}_+}\xspace}
\newcommand{\tneg}{\ensuremath{\mathcal{T}_-}\xspace}
\newcommand{\dom}{\textsf{dom}\xspace}

\newcommand{\fzt}{\textsf{fzt}\xspace}

\newcommand{\ruleTopLevel}{\textsc{top-level}\xspace}

\newcommand{\ruleVarES}{\textsc{var-\struct[]}\xspace}
\newcommand{\ruleVarAbs}{\textsc{var-abs}\xspace}
\newcommand{\ruleVarAbsTop}{\textsc{var-\l-Frozen}\xspace}
\newcommand{\ruleVarAbsBot}{\textsc{var-\l-NonFrozen}\xspace}

\newcommand{\ruleStructPhi}{\textsc{s-phi}\xspace}
\newcommand{\ruleStructOmega}{\textsc{s-omega}\xspace}
\newcommand{\ruleAppLeft}{\textsc{app-left}\xspace}
\newcommand{\ruleAppRight}{\textsc{app-right}\xspace}
\newcommand{\ruleApp}{\textsc{app}\xspace}
\newcommand{\ruleAppRightStruct}{\textsc{app-right-\struct[]}\xspace}
\newcommand{\ruleAbsTop}{\textsc{abs-$\top$}\xspace}
\newcommand{\ruleAbsBot}{\textsc{abs-$\bot$}\xspace}

\newcommand{\ruleAbsCons}{\textsc{abs-cons}\xspace}
\newcommand{\ruleAbsNil}{\textsc{abs-nil}\xspace}
\newcommand{\ruleEsLeft}{\textsc{es-left}\xspace}
\newcommand{\ruleEsLeftFrozen}{\textsc{es-left-\struct[]}\xspace}
\newcommand{\ruleEsLeftStruct}{\textsc{es-left-\struct[]}\xspace}
\newcommand{\ruleId}{\textsc{id}\xspace}
\newcommand{\ruleEs}{\textsc{es}\xspace}
\newcommand{\ruleEsRight}{\textsc{es-right}\xspace}
\newcommand{\ruleSub}{\textsc{sub}\xspace}
\newcommand{\ruleDb}{\textsc{dB}\xspace}
\newcommand{\ruleLsv}{\textsc{lsv}\xspace}
\newcommand{\ruleAuxDb}{\textsc{dB-base}\xspace}
\newcommand{\ruleAuxDbSigma}{\textsc{dB-$\sigma$}\xspace}
\newcommand{\ruleAuxLsv}{\textsc{lsv-base}\xspace}
\newcommand{\ruleAuxLsvSigma}{\textsc{lsv-$\sigma$}\xspace}
\newcommand{\ruleAuxLsvSigmaFrozen}{\textsc{lsv-$\sigma$-\struct[]}\xspace}

\newcommand{\tyruleVar}{\textsc{ty-var}\xspace}
\newcommand{\tyruleAbs}{\textsc{ty-abs}\xspace}
\newcommand{\tyruleAbsBot}{\textsc{ty-abs-$\bot$}\xspace}
\newcommand{\tyruleAbsTop}{\textsc{ty-abs-$\top$}\xspace}
\newcommand{\tyruleApp}{\textsc{ty-app}\xspace}
\newcommand{\tyruleAppStruct}{\textsc{ty-app-\struct[]}\xspace}
\newcommand{\tyruleEs}{\textsc{ty-es}\xspace}
\newcommand{\tyruleEsStruct}{\textsc{ty-es-\struct[]}\xspace}

\newcommand{\ra}{\rightarrow}
\newcommand{\ral}{\rightarrow_{\struct[\omega]}}
\newcommand{\ralf}{\rightarrow_{\struct[\varphi]}}

\newcommand\V{\mathcal{N}}

\newcommand{\lnf}[1][\varphi,\omega,\mu]{\ensuremath{\mathcal{N}_{#1}}\xspace}

\newcommand\Vewt{\V_{\varphi,\omega,\mu}}
\newcommand\Vewtt{\V_{\varphi,\omega,\top}}
\newcommand\Vewtb{\V_{\varphi,\omega,\bot}}
\newcommand\Vexwt{\V_{\varphi\cup\{x\},\omega,\mu}}
\newcommand\Vewxt{\V_{\varphi,\omega\cup\{x\},\mu}}

\newcommand\Vewxtb{\V_{\varphi,\omega\cup\{x\},\bot}}
\newcommand\Vexwtt{\V_{\varphi\cup\{x\},\omega,\top}}

\title{A strong call-by-need calculus}

\author[Th.~Balabonski]{Thibaut Balabonski}[a]
\author[A.~Lanco]{Antoine Lanco}[b]
\author[G.~Melquiond]{Guillaume Melquiond\lmcsorcid{0000-0002-6697-1809}}[b]

\address{Université Paris-Saclay, CNRS, ENS Paris-Saclay, LMF, Gif-sur-Yvette, 91190, France}
\email{thibaut.balabonski@lri.fr}

\address{Université Paris-Saclay, CNRS, ENS Paris-Saclay, Inria, LMF, Gif-sur-Yvette, 91190, France}
\email{antoine.lanco@lri.fr}
\email{guillaume.melquiond@inria.fr}


\lstdefinelanguage{Abella}{
  keywords = {Define, Theorem, kind, type},
  otherkeywords = {forall, nabla, exists},
}

\begin{document}

\begin{abstract}
\noindent
We present a \emph{call-by-need} $\lambda$-calculus that enables
\emph{strong} reduction (that is, reduction inside the body of
abstractions) and guarantees that arguments are only evaluated if needed
and at most once. This calculus uses explicit substitutions and subsumes
the existing ``strong call-by-need'' strategy, but allows for more reduction
sequences, and often shorter ones, while preserving the
\emph{neededness}.

The calculus is shown to be \emph{normalizing} in a strong sense:
Whenever a $\lambda$-term $t$ admits a normal form $n$ in the
$\lambda$-calculus, then \emph{any} reduction sequence from $t$ in the
calculus eventually reaches a representative of the normal form $n$.
We also exhibit a restriction of this calculus that has the
\emph{diamond} property and that only performs reduction sequences of
minimal length, which makes it systematically better than the existing
strategy.
We have used the Abella proof assistant to formalize part of this
calculus, and discuss how this experiment affected its design.
In particular, it led us to derive a new description of call-by-need
reduction based on inductive rules.
\end{abstract}

\maketitle

\renewcommand{\l}{\ensuremath{\lambda}\xspace}

\section{Introduction}
\label{sec:call-by-need-intro}

Lambda-calculus is seen as the standard model of computation in functional
programming languages, once equipped with an
\emph{evaluation strategy}~\cite{Plotkin-1975}.
The most famous evaluation strategies are \emph{call-by-value}, which eagerly
evaluates the arguments of a function before resolving the function call,
\emph{call-by-name}, where the arguments of a function are evaluated when
they are needed, and \emph{call-by-need}~\cite{Wadsworth, CallByNeed-95},
which extends call-by-name with a memoization or sharing mechanism to
remember the value of an argument that has already been evaluated.

The strength of call-by-name is that it only evaluates terms whose value is
effectively needed, at the (possibly huge) cost of evaluating some terms
several times.
Conversely, the strength \emph{and} weakness of call-by-value (by far the most
used strategy in actual programming languages) is that it evaluates each
function argument exactly once, even when its value is not actually needed
and when its evaluation does not terminate.
At the cost of memoization, call-by-need combines the benefits of call-by-value
and call-by-name, by only evaluating needed arguments and evaluating them only
once.

A common point of these strategies is that they are concerned with
\emph{evaluation}, that is computing \emph{values}. As such they operate in
the subset of \l-calculus called \emph{weak reduction}, in which
there is no reduction inside \l-abstractions, the latter being already
considered to be values. Some applications however, such as proof assistants or
partial evaluation, require reducing inside \l-abstractions, and possibly
aiming for the actual normal form of a \l-term.

The first known abstract machine computing the normal form of a term is
due to Crégut~\cite{Cregut-1990} and implements normal-order reduction.
More recently, several lines of work have transposed the known evaluation
strategies to strong reduction strategies or abstract machines:
call-by-value~\cite{GregoireLeroy-SCBV, BBCD-SCBV, Accattoli-SCBV},
call-by-name~\cite{Accattoli-StrongMachine},
and call-by-need~\cite{SCBN-ICFP, BC-SCBN}. Some non-advertised strong
extensions of call-by-name or call-by-need can also be found in the
internals of proof assistants, notably Coq.

These strong strategies are mostly conservative over their underlying
weak strategy, and often proceed by \emph{iteratively} applying a weak
strategy to open terms.
In other words, they use a restricted form of strong reduction to enable
reduction to normal form, but do not try to take advantage of strong
reduction to obtain shorter reduction sequences. Since call-by-need
has been shown to capture optimal weak reduction~\cite{WeakOptimality},
it is known that the deliberate use of strong reduction~\cite{FullerLaziness}
is the only way of allowing shorter reduction sequences.

This paper presents a strong call-by-need calculus, which obeys the
following guidelines. First, it only reduces needed redexes, \emph{i.e.},
redexes that are necessarily reduced in any $\beta$-reduction to normal form
of the underlying \l-term. Second, it keeps a level of sharing at least equal
to that of call-by-value and call-by-need. Third, it tries to enable strong
reduction as generally as possible.
This calculus builds on the syntax and a part of the meta-theory of
\l-calculus with explicit substitutions, which we recall in
Section~\ref{sec:lambda-c}.

We strayed away from the more traditional style of presentation of weak
call-by-need reduction, based on evaluation contexts, and went for
SOS-style reduction rules~\cite{Plotkin-SOS}, which later became a lever
for subsequent contributions of this work. Since this
style of presentation is itself a contribution, we devote
Section~\ref{sec:weak-cbn} to it.

Neededness of a redex is undecidable in general, thus the first and third
guidelines above are antagonistic. One cannot reduce only needed redexes
while enabling strong reduction as generally as possible.
Consider for instance a \l-term like $(\l{x}.t)\app(\l{y}.r)$,
where $r$ is a redex. Strong reduction could allow reducing the redex $r$
before substituting the value $\l{y}.r$, which might be particularly
interesting to prevent $r$ from being reduced multiple times after substitution.
However, we certainly do not want to do that if $r$ is a diverging term and
the normal form of $(\l{x}.t)\app(\l{y}.r)$ can be computed without
reducing $r$.

Section~\ref{sec:strong-cbn} resolves this tension by exposing a simple
syntactic criterion capturing more needed redexes than what is already
used in call-by-need strategies.
By reducing needed redexes only, our calculus enjoys a normalization
preservation theorem that is stronger than usual: Any \l-term that is
\emph{weakly} normalizing in the pure \l-calculus (\emph{i.e.}, there is at least one
reduction sequence to a normal form, but some other sequences may diverge)
will be \emph{strongly} normalizing in our calculus (\emph{i.e.}, any reduction sequence
is normalizing).
This strong normalization theorem, related to the usual \emph{completeness}
results of call-by-name or call-by-need strategies, is completely dealt with
using a system of non-idempotent intersection types. This avoids
the traditional tedious syntactic commutation lemmas, hence providing more elegant
proofs. This is an improvement over the technique used
in previous works~\cite{Kesner-CallByNeed,SCBN-ICFP}, enabled by our
SOS-style presentation.

While our calculus contains the strong call-by-need strategy introduced
in~\cite{SCBN-ICFP}, it also allows more liberal call-by-need strategies that
anticipate some strong reduction steps in order to reduce the overall
length of the reduction to normal form.
Section~\ref{sec:optimal-strategies} describes a restriction of the calculus
that guarantees reduction sequences
of minimal length.

Section~\ref{sec:abella} presents a formalization of parts of our results
in Abella~\cite{Abella14}. Beyond the proof safety provided by such
a tool, this formalization has also influenced the design of our strong
call-by-need calculus itself in a positive way. In particular, this is
what promoted the aforementioned presentation based on SOS-style local
reduction rules. The formalization can be found at the following address:
\url{https://hal.inria.fr/hal-03149692}.

Finally, Section~\ref{sec:abstract} presents an abstract machine that
implements the rules of our calculus. More precisely, it implements the
leftmost strategy that complies with our calculus, using a mutable store
to deal with explicit substitutions in a call-by-need fashion. This
section presents the big-step semantics of the machine, and draws a few
observations from it.

This paper is an extended version of~\cite{FSCD}. It contains more examples,
more explanations, and more detailed proofs than the short paper, as well as
two new sections.
First, Section~\ref{sec:weak-cbn} details the technique used for expressing
call-by-need reduction in SOS-style, and gives a proof that this presentation
actually defines the same calculus as the more traditional presentation of weak
call-by-need by means of evaluation contexts. Second, Section~\ref{sec:abstract}
on the abstract machine is fully new material. Finally, we provide in this
extended version a formal statement of how this strong call-by-need calculus
indeed reduces only needed redexes.

\section{The host calculus \texorpdfstring{\lc}{lambda-c}}
\label{sec:lambda-c}

Our strong call-by-need calculus is included in an already known calculus
\lc, that serves as a technical tool in~\cite{SCBN-ICFP} and which we name
our \emph{host calculus}.
This calculus gives the general shape of reduction rules allowing memoization
and comes with a system of non-idempotent intersection types. Its reduction,
however, is not constrained by any notion of neededness.
The \lc-calculus uses the following syntax of \l-terms with explicit
substitutions, which is isomorphic to the original syntax of the
call-by-need calculus using \letin-bindings~\cite{CallByNeed-95}.
\[ t\in\lcterms \quad::=\quad x \mid \l{x}.t \mid t\app t \mid t\esub{x}{t} \]

The free variables $\fv(t)$ of a term $t$ are defined as usual.
We call pure \l-term a term that contains no explicit substitution.
We write \ctx for a context, \emph{i.e.}, a term with exactly one hole~\hole,
and \lctx for a context with the specific shape
$\hole\esub{x_1}{t_1}\ldots\esub{x_n}{t_n}$
($n \geqslant 0$).
We write $\ctx[t]$ for the term obtained by plugging the subterm $t$ in the
hole of the context \ctx, with possible capture of free variables of $t$ by
binders in \ctx, or $t\lctx$ when the context is of the specific shape~\lctx.
We also write $\ctx[\plug{t}]$ for plugging a term $t$ whose free
variables are not captured by \ctx.

The \emph{values} we consider are \l-abstractions, and are usually
denoted by $v$. If variables were also considered as values, most of
the results of this paper would be preserved, except for the diamond
property in the restricted calculus (see \autoref{thm:diamond}).

Reduction in \lc is defined by the following three reduction rules, applied
in any context. Rather than using traditional propagation rules for explicit
substitutions~\cite{Kesner-ES}, it builds on the
\emph{Linear Substitution Calculus}~\cite{Milner-LSC, Kesner-LSC, LSC-Standardization} which is more
similar to the let-in constructs commonly used for defining call-by-need.
\[\begin{array}{rcl@{\hspace{2em}}l}
(\l{x}.t)\lctx\app u & \redb & t\esub{x}{u}\lctx \\
\ctx[\plug{x}]\esub{x}{v\lctx} & \relsv & \ctx[\plug{v}]\esub{x}{v}\lctx &
\text{with $v$ value},\\
t\esub{x}{u} & \regc & t & \text{with } x\not\in\fv(t).
\end{array}\]

The rule \redb describes $\beta$-reduction ``at a distance''. It applies to a
$\beta$-redex whose \l-abstraction is possibly hidden by a list $\lctx$ of explicit
substitutions. This rule is a combination of a single use of $\beta$-reduction
with a repeated use of the structural rule lifting the explicit
substitutions at the left of an application.
The rule \relsv describes the linear substitution of a value, \emph{i.e.}, the
substitution of one occurrence of the variable $x$ bound by an explicit
substitution. It has to be understood as a lookup operation. Similarly to
\redb, this rule embeds a repeated use of a structural rule for unnesting
explicit substitutions.
Note that this calculus only allows the substitution of \l-abstractions, and
not of variables as it is sometimes seen~\cite{MOW-CallByNeed}.
This restricted behavior is enough for the main results of this paper, and
will allow a more compact presentation.
Finally, the rule \regc describes garbage collection of an explicit
substitution for a variable that does not live anymore.
Reduction by any of these rules in any context is written $t\rec u$.

A \lc-term $t$ is related to a pure \l-term $t\unfold$ by the
unfolding operation which applies all the explicit substitutions. We
write $t\msub{x}{u}$ for the meta-level substitution of $x$ by $u$ in
$t$.
\[\begin{array}{rcl@{\hspace{3em}}rcl}
x\unfold & = & x &
(t\app u)\unfold & = & t\unfold\app u\unfold \\
(\l{x}.t)\unfold & = & \l{x}.(t\unfold) &
(t\esub{x}{u})\unfold & = & (t\unfold)\msub{x}{u\unfold}
\end{array}\]
Through unfolding, any reduction step $t\rec u$ in the \lc-calculus is related to
a sequence of reductions $t\unfold \reb^* u\unfold$ in the pure \l-calculus.

The host calculus \lc comes with a system of non-idempotent intersection
types~\cite{IntersectionTypes, Gardner-NonIdempotent}, defined
in~\cite{KesnerVentura-HW} by adding explicit substitutions to an
original system from~\cite{Gardner-NonIdempotent}.
A type~$\tau$ may be a type variable $\alpha$ or an arrow type
$\mvar\arr\tau$, where $\mvar$ is a multiset \mset{\sigma_1,\ldots,\sigma_n}
of types.

A typing environment $\Gamma$ associates to each variable in its domain
a multiset of types. This multiset contains one type for each potential use
of the variable, and may be empty if the variable is not actually used.
In the latter case, the variable can simply be removed from the typing
environment, \emph{i.e.}, we equate $\Gamma$ with $\Gamma,x:\mset{}$.

A typing judgment \jdg{\Gamma}{t}{\tau} assigns exactly one type to the term $t$.
As shown by the typing rules in \autoref{fig:base typing rules}, an argument of
an application or of an explicit substitution may be typed several times in a
derivation. Note that, in the rules, the subscript $\sigma\in\mvar$ quantifies
on all the instances of elements in the multiset $\mvar$.

\begin{figure}[t]
\begin{mathpar}
  \inferrule[\tyruleVar]{~}{\jdg{x:\mset{\sigma}}{x}{\sigma}}
  \and
  \inferrule[\tyruleApp]{\jdg{\Gamma}{t}{\mvar\arr\tau}
    \and (\jdg{\Delta_\sigma}{u}{\sigma})_{\sigma\in\mvar}}{%
    \jdg{\Gamma + \textstyle \sum_{\sigma\in\mvar} \Delta_\sigma}{t\app u}{\tau}}
  \\
  \inferrule[\tyruleAbs]{\jdg{\Gamma,x:\mvar}{t}{\tau}}{%
    \jdg{\Gamma}{\l{x}.t}{\mvar\arr\tau}}
  \and
  \inferrule[\tyruleEs]{\jdg{\Gamma,x:\mvar}{t}{\tau}
    \and (\jdg{\Delta_\sigma}{u}{\sigma})_{\sigma\in\mvar}}{%
    \jdg{\Gamma + \textstyle \sum_{\sigma\in\mvar} \Delta_\sigma}{t\esub{x}{u}}{\tau}}
\end{mathpar}
\caption{Typing rules for \lc.}
\label{fig:base typing rules}
\end{figure}

\paragraph*{Example.}
Given $\Omega=(\l{x}.x\app x)\app(\l{x}.x\app x)$,
the term $(\l{xy}.x\app (x\app x))\app (\l{z}.z)\app \Omega$ is typable in this system,
while $\Omega$ itself is not.
Let $\sigma_0$ be a type. Write $\sigma_1=\mset{\sigma_0}\arr\sigma_0$ and
$\sigma_2=\mset{\sigma_1}\arr\sigma_1$.
First we derive the judgment
\jdg{}{\l{xy}.x\app (x\app x)}{\mset{\sigma_1,\sigma_2,\sigma_2}\arr\mset{}\arr\sigma_1}
as follows.
The main subtlety of the derivation is the adequate splitting of the
environment where the application rule is used.
\begin{mathpar}
\hspace*{-16pt}
      \inferrule*[Right=\tyruleApp]{%
        \inferrule*[Left=\tyruleVar]{~}{%
          \jdg{x:\mset{\sigma_2}}{x}{\mset{\sigma_1}\arr\sigma_1}
        }
        \hspace*{5pt}
        \inferrule*[Right=\tyruleApp]{%
          \inferrule*[Left=\tyruleVar]{~}{%
            \jdg{x:\mset{\sigma_2}}{x}{\mset{\sigma_1}\arr\sigma_1}
          }
          \hspace*{13pt}
          \inferrule*[Right=\tyruleVar]{~}{%
            \jdg{x:\mset{\sigma_1}}{x}{\sigma_1}
          }
        }{%
          \jdg{x:\mset{\sigma_1,\sigma_2}}{x\app x}{\sigma_1}
        }
      \hspace*{-33pt}
      }{%
    \inferrule*[Right=\tyruleAbs]{%
        \jdg{x:\mset{\sigma_1,\sigma_2,\sigma_2},y:\mset{}}{x\app (x\app x)}{\sigma_1}
    }{%
  \inferrule*[Right=\tyruleAbs]{%
      \jdg{x:\mset{\sigma_1,\sigma_2,\sigma_2}}{\l{y}.x\app (x\app x)}{\mset{}\arr\sigma_1}
    }
  {%
    \jdg{}{\l{xy}.x\app (x\app x)}{\mset{\sigma_1,\sigma_2,\sigma_2}\arr\mset{}\arr\sigma_1}
  }
}
}
\end{mathpar}

From this first judgment, we deduce a typing derivation for
$(\l{xy}.x\app (x\app x))\app (\l{z}.z)\app \Omega$.
This derivation contains three type derivations for the first argument
$\l{z}.z$ and no type derivation for the second argument $\Omega$.
The three type derivations for $\l{z}.z$ are for each of the types in
the multiset \mset{\sigma_1,\sigma_2,\sigma_2}. Two of them are identical.
\begin{mathpar}
\hspace*{-32pt}
    \inferrule*[Right=\tyruleApp]{%
      \ldots
      \and
      \inferrule*[Right=\tyruleAbs]{%
        \inferrule*[Right=\tyruleVar]{~}{%
          \jdg{z:\mset{\sigma_0}}{z}{\sigma_0}
        }
      }{%
        \jdg{}{\l{z}.z}{\sigma_1}
      }
      \and\and
      \inferrule*[Right=\tyruleAbs]{%
        \inferrule*[Right=\tyruleVar]{~}{%
          \jdg{z:\mset{\sigma_1}}{z}{\sigma_1}
        }
      }{%
        \jdg{}{\l{z}.z}{\sigma_2}
      }
      \and\and
      \inferrule*[Right=\tyruleAbs]{%
        \inferrule*[Right=\tyruleVar]{~}{%
          \jdg{z:\mset{\sigma_1}}{z}{\sigma_1}
        }
      }{%
        \jdg{}{\l{z}.z}{\sigma_2}
      }
    }{%
  \inferrule*[Right=\tyruleApp]{%
      \jdg{}{(\l{xy}.x(xx))\app (\l{z}.z)}{\mset{}\arr\sigma_1}
  }{%
    \jdg{}{(\l{xy}.x(xx))\app (\l{z}.z)\app \Omega}{\sigma_1}
  }
}
\end{mathpar}

This type system is known to characterize \l-terms that are weakly normalizing
for $\beta$-reduction, if associated with the side condition that the empty
multiset \mset{} does not appear at a positive position in the typing judgment
$\jdg{\Gamma}{t}{\tau}$. Positive and negative occurrences of types are defined
by the following equations.
\[\begin{array}{rcl@{\hspace{3em}}rcl}
\tpos(\alpha) & = & \set{\alpha}
& \tneg(\alpha) & = & \emptyset \\
\tpos(\mvar)
& = & \set{\mvar} \cup \bigcup_{\sigma\in\mvar}\tpos(\sigma)
& \tneg(\mvar) & = & \bigcup_{\sigma\in\mvar}\tneg(\sigma) \\
\tpos(\mvar\arr\sigma) & = & \set{\mvar\arr\sigma} \cup \tneg(\mvar) \cup \tpos(\sigma)
& \tneg(\mvar\arr\sigma) & = & \tpos(\mvar) \cup \tneg(\sigma) \\
\tpos(\jdg{\Gamma}{t}{\sigma}) & = & \tpos(\sigma) \cup \bigcup_{x\in\dom(\Gamma)}\tneg(\Gamma(x))
\end{array}\]
\begin{thm}[Typability~\cite{DeCarvalho,NIIT}]\label{thm:typability}
  If the pure \l-term $t$ is weakly normalizing for $\beta$-reduction, then
  there is a typing judgment \jdg{\Gamma}{t}{\tau} such that
  $\mset{}\not\in\tpos(\jdg{\Gamma}{t}{\tau})$.
\end{thm}

\paragraph*{Example.}
Let $\alpha$, $\beta$, $\gamma$, and $\delta$ be type variables. Positive and
negative type occurrences of the type
$\tau = \mset{\mset{\alpha,\beta}\arr\gamma}\arr\mset{}\arr\delta$
are as follows:
\[\begin{array}{rcl}
\tpos(\tau) & = & \set{\mset{\mset{\alpha,\beta}\arr\gamma}\arr\mset{}\arr\delta, \mset{\alpha,\beta}, \alpha, \beta, \mset{}\arr\delta, \delta} \\
\tneg(\tau) & = & \set{\mset{\mset{\alpha,\beta}\arr\gamma}, \mset{\alpha,\beta}\arr\gamma, \gamma, \mset{}}
\end{array}\]

A typing derivation $\Phi$ for a typing judgment \jdg{\Gamma}{t}{\tau} (written
\deriv{\Gamma}{t}{\tau}) defines in $t$ a set of \emph{typed positions}, which
are the positions of the subterms of $t$ for which the derivation~$\Phi$
contains a subderivation. More precisely:
\begin{itemize}
\item $\varepsilon$ is a typed position for any derivation;
\item if $\Phi$ ends with rule \tyruleAbs, \tyruleApp or \tyruleEs, then $0p$
  is a typed position of $\Phi$ if $p$ is a typed position of the subderivation
  $\Phi'$ relative to the first premise;
\item if $\Phi$ ends with rule \tyruleApp or \tyruleEs, then $1p$
  is a typed position of $\Phi$ if $p$ is a typed position of the subderivation
  $\Phi'$ relative to one of the instances of the second premise.
\end{itemize}
Note that, in the latter case, there is no instance of the second premise and
no typed position $1p$ when the multiset \mvar is empty. On the contrary, when
\mvar has several elements, we get the union of the typed positions contributed
by each instance.

\paragraph*{Example.}
In the derivation proposed above for the judgment
\jdg{}{(\l{xy}.x(xx))\app (\l{z}.z)\app \Omega}{\sigma_1},
the position $\varepsilon$, and all the positions of the term starting with
$0$ are typed. On the contrary, no position starting with $1$ is typed.

These typed positions have an important property; they satisfy a
\emph{weighted} subject reduction theorem ensuring a decreasing derivation size,
which we will use in the next section. We call size of a derivation $\Phi$ the number of nodes of the derivation tree.
\begin{thm}[Weighted subject reduction~\cite{SCBN-ICFP}]
  \label{thm:weighted subject reduction}
  If \deriv{\Gamma}{t}{\tau} and \textup{$t\rec t'$} by reduction of a redex at a
  typed position,
  then there is a derivation \deriv[\Phi']{\Gamma}{t'}{\tau} with $\Phi'$
  smaller than $\Phi$.
\end{thm}


Note also that typing is preserved by reverting explicit substitutions
to $\beta$-redexes.
\begin{prop}[Typing and substitution]
  \jdg{\Gamma}{(\l{x}.t)\app u}{\tau}
  if and only if
  \jdg{\Gamma}{t\esub{x}{u}}{\tau}.
\end{prop}
Thus, if we write $t\revert$ the pure \l-term obtained by reverting
all the explicit substitutions of a \lc-term $t$, which can be
defined by the following equations,
\[\begin{array}{rcl@{\hspace{3em}}rcl}
x\revert & = & x &
(t\app u)\revert & = & t\revert\app u\revert \\
(\l{x}.t)\revert & = & \l{x}.(t\revert) &
(t\esub{x}{u})\revert & = & (\l{x}.t\revert)\app u\revert
\end{array}\]
then we have the following theorem.
\begin{prop}[Stability of typing by reverting explicit substitutions]
  \label{prop:revert typing}
  For any term $t$,
  \jdg{\Gamma}{t}{\tau} if and only if \jdg{\Gamma}{t\revert}{\tau}.
\end{prop}
\begin{proof}
  By induction on $t$.
\end{proof}
Given any \lc-term $t$, the reverted term $t\revert$ can be related to
the unfolded term $t\unfold$ by $\beta$-reducing precisely the $\beta$-redexes
introduced by the $\revert$ operation. This is proved by a straightforward
induction on $t$.
\begin{prop}
For any term $t$, $t\revert \reb^* t\unfold$.
\end{prop}
Hence $t\revert$ is weakly normalizing, and thus $t$ is typable, if and only
if $t\unfold$ is weakly normalizing.

\section{Weak call-by-need, revisited}
\label{sec:weak-cbn}

Call-by-need is usually defined using a formalism that is equivalent
or equal to \lc, by restricting the situations in which the reduction rules
can be applied.
The main idea is that since the argument of a function is not always needed,
we do not reduce in advance the right part of an application $t\app u$.
Instead, we first evaluate $t$ to an answer $(\l{x}.t')\lctx$ (a value under
some explicit substitutions), then apply a
\rdb-reduction to put the argument $u$ in the environment of $t'$, and
then go on with the resulting term $t'\esub{x}{u}\lctx$, evaluating $u$ only if
and when it is required.

\paragraph{Call-by-need using evaluation contexts}

This restriction on the reduced redexes is traditionally defined in
weak call-by-need through a notion of evaluation
context~\cite{CallByNeed-95,CallByNeed-2012}.
Such a definition using the syntax of our host calculus \lc can for instance
be found in~\cite{Accattoli-AbstractMachines}.
Evaluation contexts \ectx are characterized by the following grammar:
\[\begin{array}{rcl}
\ectx & ::= & \hole \mid \ectx\app u \mid \ectx\esub{x}{u} \mid \ectx[\plug{x}]\esub{x}{\ectx}
\end{array}\]

In this presentation, the \rdb rule is used exactly as it is presented in \lc,
and the \rlsv rule is restricted to allow substitution only for an occurrence
of $x$ that is reachable through an evaluation context. The rule \rgc is not
used, since it is never actually necessary.
\[\begin{array}{rcl@{\hspace{2em}}l}
(\l{x}.t)\lctx\app u & \mapsto_\rdb & t\esub{x}{u}\lctx \\
\ectx[\plug{x}]\esub{x}{v\lctx} & \mapsto_\rlsv & \ectx[\plug{v}]\esub{x}{v}\lctx
& \text{with $v$ value.}
\end{array}\]

Then $t$ is defined to reduce to $t'$ for a base rule $\rho$ whenever there
exists an evaluation context \ectx and a $\rho$-redex $r$ such that
$t=\ectx[r]$, $r\mapsto_\rho r'$ and $t'=\ectx[r']$.
Let us call \lwcbn the weak call-by-need calculus defined by these rules and
contexts.

\paragraph*{Example.}
If $v$ is a value, then
$(y\app u)\esub{y}{v\esub{z}{w}}\mapsto_\rlsv(v\app u)\esub{y}{v}\esub{z}{w}$
is an instance of the base reduction rule \rlsv. Indeed,
$(y\app u)\esub{y}{v\esub{z}{w}}$ can be written $\ectx[\plug{y}]\esub{y}{v\lctx}$
with $\ectx=\hole\app u$ an evaluation context and $\lctx=\esub{z}{w}$ a list of
explicit substitutions.
Moreover $\ecx'(\hole)=(x\app t)\esub{x}{\hole}$ is an evaluation context,
built with the shape $\ecx'_1\fplug{x}\esub{x}{\ecx'_2(\hole)}$ with $\ecx'_1(\hole)=\hole\app t$
and $\ecx'_2(\hole)=\hole$.
Thus by applying the previous $\mapsto_\rlsv$-reduction in the evaluation
context $\ecx'(\hole)$, we justify the following reduction step:
\[ (x\app t)\esub{x}{(y\app u)\esub{y}{v\esub{z}{w}}}
\re (x\app t)\esub{x}{(v\app u)\esub{y}{v}\esub{z}{w}} \]

This presentation based on evaluation contexts is difficult to handle.
A notable issue is the formalization of the plugging operations $\ectx[t]$
and $\ectx[\plug{t}]$ and their relation with $\alpha$-equivalence.
Hence we propose here a new presentation of weak call-by-need reduction
based on an equivalent inductive definition, which we will extend in the rest
of the paper to define and formalize strong call-by-need reduction.

\paragraph{Inductive definition of call-by-need reduction}

We provide an alternative definition of~\lwcbn, based on an inductive
definition of a binary relation $\rex{\rho}$ on terms, parameterized by a
rule $\rho$ which can be \rdb, \rlsv, or two others that we will introduce
shortly. The inference rules are given in \autoref{fig:wcbn-rules}.

Rule \ruleAppLeft makes reduction always possible on the left of an
application, and rule \ruleEsLeft does the same for terms that are under
an explicit substitution. They correspond respectively to the shapes
$\ectx\app u$ and $\ectx\esub{x}{u}$ of evaluation contexts.
Remark that there is no inference rule allowing reduction inside a
\l-abstraction. Therefore, this truly defines only weak reduction steps.

Rule \ruleEsRight restricts reduction of the argument of a substitution to the
case where an occurrence of the variable affected by the substitution is at a
reducible position, as did the evaluation context shape
$\ectx[\plug{x}]\esub{x}{\ectx}$.
In order to probe a term for the presence of some variable $x$ at a reducible
position, which is what the condition $\ectx[\plug{x}]$ means, without
mentioning any evaluation context, it uses an auxiliary rule \ri{x} which
propagates using the same inductive reduction relation. This auxiliary
reduction rule does not modify the term to which it applies, as witnessed by
its base case \ruleId.

Rules \ruleDb and \ruleLsv are the base cases for applying reductions \rdb or
\rlsv. Using the context notations, they allow the following reductions.
\[\begin{array}{rcl@{\hspace{2em}}l}
(\l{x}.t)\lctx\app u & \rex{\rdb} & t\esub{x}{u}\lctx \\
\ectx[\plug{x}]\esub{x}{v\lctx} & \rex{\rlsv} & \ectx[\plug{v}]\esub{x}{v}\lctx
& \text{with $v$ value.}
\end{array}\]

However, we are aiming for a definition that does not rely on the notion
of context. We have to express \rlsv-reduction without the evaluation
context \ectx.
Moreover, a list \lctx of explicit substitutions \emph{is} a context and
thus brings the same formalization problems as evaluation contexts.
Therefore we define each of these rules using an auxiliary
(inductive) reduction relation dealing with the list \lctx of explicit
substitutions and, in the case of \rlsv, with the evaluation context \ectx.
The inference rules for these auxiliary relations are given in
\autoref{fig:waux-rules}.

\begin{figure}[t]
  \begin{mathpar}
  \inferrule[\ruleAppLeft]{t \rex{\rho} t'}{%
    t\app u \rex{\rho} t'\app u}
  \and
  \inferrule[\ruleEsLeft]{%
    t \rex{\rho} t'}{%
    t\esub{x}{u} \rex{\rho} t'\esub{x}{u}}
  \and
  \inferrule[\ruleEsRight]{%
    t \rex{\ri{x}} t%
    \and u \rex{\rho} u'}{%
    t\esub{x}{u} \rex{\rho} t\esub{x}{u'}}
  \\
  \inferrule[\ruleId]{~}{x \rex{\ri{x}} x}
  \and
  \inferrule[\ruleSub]{~}{x \rex{\rs{x}{v}} v}
  \and
  \inferrule[\ruleDb]{t \auxdb t'}{t \rex{\rdb} t'}
  \and
  \inferrule[\ruleLsv]{%
    t \wauxlsv t'}{%
    t \rex{\rlsv} t'}
  \end{mathpar}
  \caption{Reduction rules for \lwcbn.}
  \label{fig:wcbn-rules}
\end{figure}

\begin{figure}[t]
  \begin{mathpar}
  \inferrule[\ruleAuxDb]{~}{(\l{x}.t)\app u \auxdb t\esub{x}{u}}
  \and
  \inferrule[\ruleAuxLsv]{%
    t \rex{\rs{x}{v}} t'%
    \and v\text{ value}}{%
    t\esub{x}{v} \wauxlsv t'\esub{x}{v}}
  \\
  \inferrule[\ruleAuxDbSigma]{t\app u \auxdb s}{%
    t\esub{x}{w}\app u \auxdb s\esub{x}{w}}
  \and
  \inferrule[\ruleAuxLsvSigma]{%
    t\esub{x}{u} \wauxlsv t'}{%
    t\esub{x}{u\esub{y}{w}} \wauxlsv t'\esub{y}{w}}
  \end{mathpar}
  \caption{Auxiliary reduction rules for \lwcbn.}
  \label{fig:waux-rules}
\end{figure}

Rules \ruleAuxDb and \ruleAuxLsv describe the base cases of the auxiliary
reductions, where the list \lctx is empty. Note that, while \ruleAuxDb is an
axiom, the inference rule \ruleAuxLsv uses as a premise a reduction
\rex{\rho} using a new reduction rule $\rho = \rs{x}{v}$.
This rule substitutes one occurrence of the variable $x$ at
a reducible position by the value~$v$, and implements the transformation of
$\ectx[\plug{x}]$ into $\ectx[\plug{v}]$. As seen for \ri{x} above,
this reduction rule propagates using the same inductive reduction relation,
and its base case is the rule \ruleSub in \autoref{fig:wcbn-rules}.

Rule \ruleAuxDbSigma makes it possible to float out an explicit substitution
applied to the left part of an application. That is, if a \rdb-reduction is
possible without the substitution, then the reduction is performed and the
substitution is applied to the result.
Rule \ruleAuxLsvSigma achieves the same effect with the nested substitutions
applied to the value substituted by an \rlsv-reduction step.
In other words, each of these two rules applied repeatedly moves an arbitrary
list \lctx of substitutions until the corresponding base case becomes
applicable.

Finally, we say $t$ reduces to $t'$ whenever we can derive
$t\rex{\rho}t'$ with $\rho\in\set{\rdb,\rlsv}$.

\paragraph*{Example.}
The reduction $(x\app t)\esub{x}{(y\app u)\esub{y}{v\esub{z}{w}}}
\rex{\rlsv} (x\app t)\esub{x}{(v\app u)\esub{y}{v}\esub{z}{w}}$
presented in the previous example is also derivable with these
reduction rules. The derivation contains two main branches.
The left branch checks that an occurrence of $x$ is reachable in the
main term $x\app t$, using the relation~$\rex{\ri{x}}$.
The right branch then performs reduction in the argument of the
substitution. This right branch immediately invokes the auxiliary reduction
\wauxlsv, to conclude in two steps: first moving out the substitution
$\esub{z}{w}$ which hides the value $v$, using a $\sigma$-rule, and then
applying a linear substitution of $y$ by $v$ in the term $y\app u$, using
the relation~$\rex{\rs{y}{v}}$.
\begin{mathpar}
  \inferrule*[Right=\ruleEsRight]{%
    \inferrule*[Left=\ruleAppLeft]{%
      \inferrule*[Left=\ruleId]{~}{x\rex{\ri{x}}x}
    }{x\app t\rex{\ri{x}} x\app t}
    \and\and
    \inferrule*[Right=\ruleLsv]{%
      \inferrule*[Right=\ruleAuxLsvSigma]{%
        \inferrule*[Right=\ruleAuxLsv]{%
          \inferrule*[Left=\ruleAppLeft]{%
            \inferrule*[Left=\ruleSub]{~}{y\rex{\rs{y}{v}}v}
          }{
            y\app u\rex{\rs{y}{v}} v\app u
          }
          \and
          v\text{ value}
        }{
          (y\app u)\esub{y}{v}
          \wauxlsv (v\app u)\esub{y}{v}
        }
      }{
        (y\app u)\esub{y}{v\esub{z}{w}}
        \wauxlsv (v\app u)\esub{y}{v}\esub{z}{w}
      }
    }{
      (y\app u)\esub{y}{v\esub{z}{w}}
      \rex{\rlsv} (v\app u)\esub{y}{v}\esub{z}{w}
    }
  }{
    (x\app t)\esub{x}{(y\app u)\esub{y}{v\esub{z}{w}}}
    \rex{\rlsv} (x\app t)\esub{x}{(v\app u)\esub{y}{v}\esub{z}{w}}
  }
\end{mathpar}

\begin{thm}[Equivalence]\label{thm:equivalence}
  We have the following equivalences between the inductive definition of
  \lwcbn and the evaluation context-based presentation.
  \begin{enumerate}
  \item $t\rex{\rdb}t'$ iff there are \ectx, $r$, and $r'$ such that
    $t=\ectx[r]$, $r\mapsto_\rdb r'$, and $t'=\ectx[r']$;
  \item $t\rex{\rlsv}t'$ iff there are \ectx, $r$, and $r'$ such that
    $t=\ectx[r]$, $r\mapsto_\rlsv r'$, and $t'=\ectx[r']$;
  \item $t\rex{\ri{x}}t$ iff there is \ectx such that $t=\ectx[\plug{x}]$;
  \item $t\rex{\rs{x}{v}}t'$ iff there is \ectx such that $t=\ectx[\plug{x}]$
    and $t'=\ectx[\plug{v}]$.
  \end{enumerate}
\end{thm}

The proof of this theorem also requires statements relating the auxiliary
reduction relations \auxdb and \wauxlsv to the base reduction rules
$\mapsto_\rdb$ and $\mapsto_\rlsv$.
\begin{lem}[Equivalence]\label{lem:equivalence}
  We have the following equivalences between the inductive definition of
  the auxiliary rules of \lwcbn and the base reduction rules of the
  evaluation context-based presentation.
  \begin{enumerate}
  \item \label{lem:equivalence:db}
    $t\auxdb t'$ iff $t\mapsto_\rdb t'$;
  \item \label{lem:equivalence:lsv}
    $t\wauxlsv t'$ iff $t\mapsto_\rlsv t'$.
  \end{enumerate}
\end{lem}
The second equivalence has to be proved in a mutual induction with the
main theorem, since the inference rules for \wauxlsv use \rex{\rho} and
the base rule $\mapsto_\rlsv$ uses an evaluation context.
The first equivalence however is a standalone property.
\begin{proof}[Proof of \autoref{lem:equivalence} \autoref{lem:equivalence:db}.]
  The implication from left to right is a simple induction on the derivation
  of $t\auxdb t'$, while the reverse implication is proved by induction on the
  length of the list \lctx.
\end{proof}

\begin{proof}[Proof of \autoref{thm:equivalence} and \autoref{lem:equivalence} \autoref{lem:equivalence:lsv}.]
  The five equivalences are proved together. Implications from left to right
  are proved by a single induction on the derivation of $t\rex{\rho}t'$
  or $t\wauxlsv t'$, whereas implications from right to left are proved by a
  single induction on the evaluation context \ectx and on the length of the
  list \lctx.
  In the detailed proof, we extend the notation $\mapsto_\rho$
  for any $\rho$, by defining $x\mapsto_{\ri{x}}x$ and $x\mapsto_{\rs{x}{v}}v$.

  \noindent
  \emph{Implications from left to right, by mutual induction.}
  Consider a step \rex{\rho}.
  \begin{itemize}
  \item Case $t\app u\rex{\rho} t'\app u$ by rule \textsc{app-left}
    with $t\rex{\rho}t'$.
    By induction hypothesis, there are an evaluation context \ectx
    and two terms $r$ and $r'$, such that
    $t=\ectx[r]$, $r\mapsto_\rho r'$, and $t'=\ectx[r']$.
    We conclude using the context $\ecx' = \ectx\app u$, which is an
    evaluation context.
  \item Case $t\esub{x}{u}\rex{\rho} t'\esub{x}{u}$ by
    rule \textsc{es-left}
    is identical.
  \item Case $t\esub{x}{u}\rex{\rho}t\esub{x}{u'}$ by rule
    \textsc{es-right} with $t\rex{\ri{x}}t$ and $u\rex{\rho}u'$.
    By induction hypotheses, we have $\ecx_1$ such that $t=\ecx_1(\plug{x})$
    and $\ecx_2$, $r$, and $r'$, such that $u=\ecx_2(r)$, $r\mapsto_\rho r'$,
    and $u'=\ecx_2(r')$.
    We conclude using the context $\ecx' = \ecx_1(\plug{x})\esub{x}{\ecx_2(\hole)}$,
    which is an evaluation context.
  \item Cases $x\rex{\ri{x}}x$ and $x\rex{\rs{x}{v}}v$ by rules
    \textsc{id} and \textsc{sub} are immediate with the empty context.
  \item Case $t\rex{\rdb}t'$ by rule \textsc{dB} with $t\auxdb t'$ follows from
    \autoref{lem:equivalence} \autoref{lem:equivalence:db}.
  \item Case $t\rex{\rlsv}t'$ by rule \textsc{lsv} with $t\wauxlsv t'$ follows from
    \autoref{lem:equivalence} \autoref{lem:equivalence:lsv}, by
    mutual induction.
  \end{itemize}
  Consider a step \wauxlsv.
  \begin{itemize}
  \item Case $t\esub{x}{v} \wauxlsv t'\esub{x}{v}$ by rule \textsc{lsv-base}
    with $t \rex{\rs{x}{v}} t'$ and $v$ value.
    By (mutual) induction hypothesis, there is \ecx such that
    $t = \ectx[\plug{x}]$ and $t' = \ectx[\plug{v}]$.
    Then we have $\ectx[\plug{x}]\esub{x}{v} \mapsto_\rlsv \ectx[\plug{v}]\esub{x}{v}$.
  \item Case $t\esub{x}{u\esub{y}{w}} \wauxlsv t'\esub{y}{w}$ by
    rule \textsc{lsv-$\sigma$} with $t\esub{x}{u} \wauxlsv t'$.
    By induction hypothesis, $t\esub{x}{u} \mapsto_\rlsv t'$.
    Then necessarily there are \ecx, $v$, and \lctx, with $v$ a value, such that
    $t=\ectx[\plug{x}]$, $u=v\lctx$, and $t'=\ectx[\plug{v}]\esub{x}{v}\lctx$.
    Then $t\esub{x}{u\esub{y}{w}} = \ectx[\plug{x}]\esub{x}{v\lctx\esub{y}{w}}
    \mapsto_\rlsv \ectx[\plug{v}]\esub{x}{v}\lctx\esub{y}{w} = t'\esub{y}{w}$.
  \end{itemize}

  \noindent
  \emph{Implications from right to left, by mutual induction.}
  Consider a context and a reduction $r\mapsto_\rho r'$.
  \begin{itemize}
  \item Case $\ecx = \hole$. By case on $\rho$.
    \begin{itemize}
    \item If $r\mapsto_\rdb r'$ then $\ectx[r]=r\rex{\rdb}r'=\ectx[r']$
      by \autoref{lem:equivalence} \autoref{lem:equivalence:db}.
    \item If $r\mapsto_\rlsv r'$ then $\ectx[r]=r\rex{\rlsv}r'=\ectx[r']$
      by mutual induction hypothesis.
    \item If $r\mapsto_{\ri{x}}r'$ or $r\mapsto_{\rs{x}{v}}r'$ then
      conclusion is by a base rule of \rex{\rho}.
    \end{itemize}
  \item Case $\ectx\app u$.
    By induction hypothesis, $\ectx[r]\rex{\rho}\ectx[r']$, and
    then by rule \textsc{app-left}, $\ectx[r]\app u\rex{\rho}\ectx[r']\app u$.
  \item Case $\ectx\esub{x}{u}$ is identical.
  \item Case $\ecx_1(\plug{x})\esub{x}{\ecx_2(\hole)}$.
    By induction hypotheses, $\ecx_1(\plug{x}) \rex{\ri{x}} \ecx_1(\plug{x})$
    and $\ecx_2(r) \rex{\rho} \ecx_2(r')$, and we conclude by rule
    \textsc{es-right}.
  \end{itemize}
  Consider a step $\mapsto_\rlsv$.
  \begin{itemize}
  \item Case $\ectx[\plug{x}]\esub{x}{v} \mapsto_\rlsv \ectx[\plug{v}]\esub{x}{v}$.
    Conclusion is immediate by mutual induction hypothesis.
  \item Case $\ectx[\plug{x}]\esub{x}{v\lctx\esub{y}{w}} \mapsto_\rlsv \ectx[\plug{v}]\esub{x}{v}\lctx\esub{y}{w}$.
    By induction hypothesis,
    $\ectx[\plug{x}]\esub{x}{v\lctx} \wauxlsv \ectx[\plug{v}]\esub{x}{v}\lctx$,
    and we conclude by rule \textsc{lsv-$\sigma$}. \qedhere
  \end{itemize}
\end{proof}

\section{Strong call-by-need calculus \texorpdfstring{\lscbn}{lambda-sn}}
\label{sec:strong-cbn}

Our strong call-by-need calculus is defined by the same terms and reduction
rules as~\lc, with restrictions on where the reduction rules can be applied.
These restrictions ensure in particular that only the needed redexes are
reduced. Notice that \rgc-reduction is never needed in this calculus
and will thus be ignored from now on.

\subsection{Reduction in \texorpdfstring{\lscbn}{lambda-sn}}
\label{sec:red-lscbn}

The main reduction relation is written $t\red t'$ and represents
one step of \rdb- or \rlsv-reduction at an eligible position of the term $t$.
The starting point is the inductive definition of weak call-by-need
reduction, which we now extend to account for strong reduction.

\paragraph*{Top-level-like positions and frozen variables.}
Strong reduction brings new behaviors that cannot be observed in weak
reduction. For instance, consider the top-level
term $\l{x}.x\app t\app u$. It is an abstraction, which would not need to be
further evaluated in weak call-by-need. Here however, we have to reduce it
further to reach its putative normal form. So, let us gather some knowledge
on the term. Given its position, we know that this abstraction will
never be applied to an argument. This means in particular that its variable
$x$ will never be substituted by anything; it is blocked and is now part
of the rigid structure of the term. Following~\cite{SCBN-ICFP}, we say that
variable $x$ is \emph{frozen}. As for the arguments $t$ and $u$ given to
the frozen variable~$x$, they will always remain at their respective
positions and their neededness is guaranteed. So, the calculus allows
their reduction. Moreover, these subterms $t$ and $u$ will never be applied to
other subterms; they
are in \emph{top-level-like} positions and can be treated as independent
terms. In particular, assuming that the top-level term is $\l{x}.x\app (\l{y}.t')\app u$
(that is, $t$ is the abstraction~$\l{y}.t'$), the variable $y$
will never be substituted and both variables $x$ and $y$ can be seen as
frozen in the subterm $t'$.

Let us now consider the top-level term $(\l{x}.x\app (\l{y}.t')\app u)\app v$,
\emph{i.e.}, the previous one applied to some argument $v$. The analysis becomes
radically different. Indeed, both abstractions in this term are at positions
where they may eventually interact with other parts of the term:
$(\l{x}\ldots)$ is already applied to an argument, while $(\l{y}.t')$ might
eventually be substituted at some position inside $v$ whose properties are
not yet known.
Thus, none of the abstractions is at a top-level-like position and we cannot
rule out the possibility that some occurrences of $x$ or $y$ become substituted
eventually. Consequently, neither $x$ nor $y$ can be considered as frozen.
In addition, notice that the subterms $\l{y}.t'$ and $u$ are not even
guaranteed to be needed in $(\l{x}.x\app (\l{y}.t')\app u)\app v$.
Thus our calculus shall prevent them from being reduced until it gathers
more information about $v$.

Therefore, the positions of a subterm $t$ that are eligible for
reduction largely depend on the context surrounding $t$. They depend in
particular on the set of variables that are frozen by this context, which is
itself mutually dependent with the notion of top-level-like positions.
To characterize top-level-like positions and frozen variables, we use a
judgment $\ctx\vdash\mu,\varphi$ where
$\ctx$ is a context,
$\mu$ is a boolean flag whose value can be $\top$ or $\bot$, and
$\varphi$ is a set of variables.
A context hole \ctx is in a top-level-like position
if one can derive $\ctx\vdash\top,\varphi$. A variable $x$ is frozen in
the context \ctx if one can derive $\ctx\vdash\mu,\varphi$ with $x\in\varphi$.

\begin{figure}[t]
\begin{mathpar}
  \inferrule[\ruleTopLevel]
    {~}
    {\bullet \vdash \top,\varphi}
    \and
  \inferrule[\ruleAbsTop]
    {\ctx \vdash \top,\varphi}
    {\ctx[ \lambda x. \bullet ] \vdash \top,\varphi \cup \{x\}}
    \and
  \inferrule[\ruleAbsBot]
    {\ctx \vdash \bot,\varphi}
    {\ctx[ \lambda x. \bullet ] \vdash \bot,\varphi}
    \\
  \inferrule[\ruleAppLeft]
    {\ctx \vdash \mu,\varphi}
    {\ctx[ \bullet\app t ] \vdash \bot,\varphi}
    \and
  \inferrule[\ruleAppRight]
    {\ctx \vdash \mu,\varphi}
    {\ctx[ t\app\bullet ] \vdash \bot,\varphi}
    \and
  \inferrule[\ruleAppRightStruct]
    {\ctx \vdash \mu,\varphi \and t\in \struct[\varphi]}
    {\ctx[ t\app\bullet ] \vdash \top,\varphi}
    \\
  \inferrule[\ruleEsLeft]
    {\ctx \vdash \mu,\varphi}
    {\ctx[ \bullet\esub{x}{u} ] \vdash \mu,\varphi}
    \and
  \inferrule[\ruleEsRight]
    {\ctx \vdash \mu,\varphi}
    {\ctx[ t\esub{x}{\bullet} ] \vdash \bot,\varphi}
    \and
  \inferrule[\ruleEsLeftStruct]
    {\ctx \vdash \mu,\varphi \and u\in \struct[\varphi]}
    {\ctx[ \bullet\esub{x}{u} ] \vdash \mu,\varphi\cup \{x\}}
\end{mathpar}
\caption{Top-level-like positions and frozen variables.}
\label{fig:frozen}
\end{figure}

\autoref{fig:frozen} gives an inductive definition of the judgment
$\ctx\vdash\mu,\varphi$.
A critical aspect of this definition is the direction of the flow of
information. Indeed, we deduce information on the hole of the context from the
context itself. That is, information is flowing outside-in: from top-level
toward the hole of the context.

Rule \ruleTopLevel states that the top-level position of a term is
top-level-like, and that any (free) variable can be considered to be frozen
at this point.

Rules \ruleAbsTop and \ruleAbsBot indicate that a variable $x$ bound by an
abstraction is considered to be frozen if this abstraction is at a
top-level-like position. Moreover, if an abstraction is at a top-level-like
position, then the position of its body is also top-level-like.

Notice that frozen variables in a term $t$ are either free variables of $t$,
or variables introduced by binders in $t$. As such they obey the usual renaming
conventions. In particular, the rules \ruleAbsTop and \ruleAbsBot
implicitly require that the variable $x$ bound by the abstraction
is fresh and hence \emph{not} in the set $\varphi$. We keep this
\emph{freshness convention} in all the definitions of the paper.
In particular, it also applies to the three rules dealing with explicit
substitution.

Rules \ruleAppLeft, \ruleAppRight, \ruleEsLeft, and \ruleEsRight state that
the variables frozen on either side of an application or an explicit
substitution are the variables that are frozen at the level of the application
or explicit substitution itself. Rule \ruleEsLeft indicates in addition that
the top-level-like status of a position is inherited on the left side of
an explicit substitution.

Rule \ruleAppRightStruct strengthens the rule \ruleAppRight
in the specific case where the application is led by a frozen variable.
When this criterion is met, the argument position of an application can be
considered as top-level-like.
Rule \ruleEsLeftStruct strengthens the rule \ruleEsLeft when the argument
of an explicit substitution is led by a frozen variable. In this case,
the variable of the substitution can itself be considered as frozen.

\begin{figure}[t]
  \begin{mathpar}
    \inferrule{x\in\varphi}{x\in\struct}
    \and
    \inferrule{t\in\struct}{t\app u\in\struct}
    \and
    \inferrule{t\in\struct}{t\esub{x}{u}\in\struct}
    \and
    \inferrule{t\in\struct[\varphi\cup\set{x}] \and u\in\struct}{t\esub{x}{u}\in\struct}
  \end{mathpar}
  \caption{Structures of \lscbn.}
  \label{fig:structures}
\end{figure}

The specific criterion of the last two rules is made formal through the
notion of \emph{structure}, which is an application
$x\app t_1\app \ldots\app t_n$ led by a frozen variable $x$, possibly
interlaced with explicit substitutions (\autoref{fig:structures}).
We write \struct the set of structures under a set $\varphi$ of frozen
variables.
It differs from the notion in~\cite{SCBN-ICFP} in that it does not
require the term to be in normal form. For example, we have $x\app t
\in \struct[\{x\}]$ even if $t$ still contains some redexes.

The reason for considering structures in rules \ruleAppRightStruct and
\ruleEsLeftStruct is that they cannot reduce to an abstraction, even
under an explicit substitution. For example, $x\esub{x}{\lambda z.t}$
is not a structure, whatever the set $\varphi$. See also
\autoref{lem:stab-struct} for another take on structures.

Notice that, by our freshness condition, it is assumed that $x$ is
\emph{not} in the set $\varphi$ in the third and fourth rules of
\autoref{fig:structures}.
As implied by the last rule in \autoref{fig:structures}, an explicit
substitution in a structure may even affect the leading variable,
provided that the content of the substitution is itself a structure,
\emph{e.g.}, $(x\app t)\esub{x}{y\app u} \in \struct[\{y\}]$.

\paragraph*{Example.}
Consider the term $\l{x}.((\l{z}.x\app z\app (\l{y}.y\app t)\app(z\app u))\app v)$,
where $t$, $u$, and $v$ are arbitrary, not necessarily closed, subterms.
The above rules allow us to derive that the subterm
$x\app z\app(\l{y}.y\app t)\app(z\app u)$
is at a position where $x$ is a frozen variable.
\begin{mathpar}
  \inferrule*[Right=\ruleAbsBot]{%
    \inferrule*[Right=\ruleAppLeft]{%
      \inferrule*[Right=\ruleAbsTop]{%
        \inferrule*[Right=\ruleTopLevel]{~}{%
          \hole\vdash\top,\emptyset
        }
      }{%
        \l{x}.\hole\vdash\top,\set{x}
      }
    }{%
      \l{x}.(\hole\app v)\vdash\bot,\set{x}
    }
  }{%
    \l{x}.((\l{z}.\hole)\app v)\vdash\bot,\set{x}
  }
\end{mathpar}

This position, however, is \emph{not} top-level-like, as it is the position of
the body of an applied \l-abstraction. This can be observed in the previous
derivation by the unavoidable application of the rule \ruleAppLeft.
Independently of the fact that its position is not top-level-like, the subterm
$x\app z\app(\l{y}.y\app t)\app(z\app u)$ is a structure, as it is led by the
frozen variable $x$.
We can then build a follow-up of this derivation, first deriving with a
combination of the rules \ruleAppLeft and \ruleAppRightStruct that the second
argument $\l{y}.y\app t$ of this structure \emph{is} at a top-level-like
position, then deducing that $y$ is frozen and that the subterm $t$ is at a
top-level-like position.
\begin{mathpar}
  \inferrule*[Left=\ruleAppRightStruct]{%
    \inferrule*[Left=\ruleAbsTop]{%
      \inferrule*[Left=\ruleAppRightStruct]{%
        \inferrule*[Left=\ruleAppLeft]{%
          \l{x}.((\l{z}.\hole)\app v)\vdash\bot,\set{x}
        }{%
          \l{x}.((\l{z}.\hole\app(z\app u))\app v)\vdash\bot,\set{x}
        }
        \and
        x\app z\in\struct[\set{x}]
      }{%
        \l{x}.((\l{z}.(x\app z\app\hole)\app(z\app u))\app v)\vdash\top,\set{x}
      }
    }{%
      \l{x}.((\l{z}.x\app z\app(\l{y}.\hole)\app(z\app u))\app v)\vdash\top,\set{x,y}
    }
    \and
    y\in\struct[\set{x,y}]
  }{%
    \l{x}.((\l{z}.x\app z\app(\l{y}.y\app \hole)\app(z\app u))\app v)\vdash\top,\set{x,y}
  }
\end{mathpar}

Conversely, remark that the subterm $z\app u$, which it is at a top-level-like
position, is not a structure since $z$ is not frozen. Thus, $u$ is not at a
top-level-like position.

\paragraph*{Inductive definition of strong call-by-need reduction.}
The definition of reduction in our strong call-by-need calculus extends the
inductive definition given in Section~\ref{sec:weak-cbn}. First of all,
reductions have to be allowed, at least in some circumstances,
inside \l-abstractions and on the right side of an application. The criteria
we apply are:
\begin{itemize}
\item reduction under a needed \l-abstraction is always allowed, and
\item reduction on the right side of an application is allowed if and only
  if this application is both needed and a structure (\emph{i.e.,} led by
  a frozen variable).
\end{itemize}

Applying the latter criterion requires identifying the frozen variables, and
this identification in turn requires keeping track of the
top-level-like positions. For this, we define an inductive reduction relation
$t\rescbn{\rho}{\varphi}{\mu}t'$ parameterized by a reduction rule $\rho$ and by
some context information $\varphi$ and $\mu$. This relation plays two roles: identifying a
position where a reduction rule can be applied in $t$, and performing said
reduction.
The information on the context of the subterm $t$ mirrors
exactly the information given by a judgment $\ctx\vdash\mu,\varphi$. We thus
have:
\begin{itemize}
\item a flag $\mu$ indicating whether $t$ is at a top-level-like
  position ($\top$) or not ($\bot$);
\item the set $\varphi$ of variables that are frozen at the considered
  position.
\end{itemize}

Consequently, the top-level reduction relation $t\red t'$ holds whenever $t$
reduces to $t'$ by rule \rdb or \rlsv in the empty context. In other words,
$t\rescbn{\rdb}{\varphi}{\mu}t'$ or $t\rescbn{\rlsv}{\varphi}{\mu}t'$,
where the flag $\mu$ is $\top$, and the set $\varphi$ is typically empty
when $t$ is closed, or contains the free variables of $t$ otherwise.

The inference rules for $\rescbn{\rho}{\varphi}{\mu}$ are given in
\autoref{fig:scbn-rules}.
Notice that $\rho$ flows inside-out, \emph{i.e.}, from the position of the
reduction itself to top-level, as it was already the case in the definition
of weak call-by-need reduction (\autoref{fig:wcbn-rules}).
On the contrary, information about $\varphi$ and~$\mu$ flows outside-in, that
is from top-level to the position of the reduction. This also reflects the
information flow already seen in the definition of the judgment
$\ctx\vdash\mu,\varphi$. However, this now appears has an \emph{upward} flow
of information in the inference rules, which may seem like a reversal.
This comes from the fact that the inference rules in \autoref{fig:scbn-rules}
focus on terms, while \autoref{fig:frozen} focused on contexts.
Notice also that we do not rely explicitly on the judgment
$\ctx\vdash\mu,\varphi$ in the inference rules. Instead, we blend the
inference rules defining this judgment into the reduction rules.

Rule \ruleAppLeft makes reduction always possible on the
left of an application, but as shown by the premise, this position is not a
$\top$ position. Rule \ruleAppRight on the other hand allows reducing
on the right of an application, and even doing so in $\top$ mode, but only
when the application is led by a frozen variable. The latter is expressed
using the notion of structure defined above (\autoref{fig:structures}).

Rules \ruleAbsTop and \ruleAbsBot make reduction always possible
inside a \l-abstraction, \emph{i.e.}, unconditional strong reduction.
If the abstraction is in a $\top$ position,
its variable is added to the set of frozen variables while reducing
the body of the abstraction.
Rules \ruleEsLeft and \ruleEsLeftFrozen show that it is always possible to
reduce a term affected by an explicit substitution. If the
substitution contains a structure, the variable bound by the
substitution can be added to the set of frozen variables.
Rule \ruleEsRight restricts reduction inside a substitution to the case where
an occurrence of the substituted variable is at a reducible position.
It uses the auxiliary rule \ri{x} already discussed in Section~\ref{sec:weak-cbn},
which propagates using the same inductive reduction relation, to probe a term
for the presence of some variable~$x$ at a reducible position. By freshness,
in all the rules with a binder, the bound variable $x$ can appear neither in $\varphi$
nor in $\rho$, that is, if $\rho = \ri{y}$ then
$x\neq y$, and if $\rho = \rs{y}{v}$ then $x\neq y$ and $x\not\in\fv(v)$.

\begin{figure}[t]
  \begin{mathpar}
  \inferrule[\ruleAppLeft]{t \rescbn{\rho}{\varphi}{\bot} t'}{%
    t\app u \rescbn{\rho}{\varphi}{\mu} t'\app u}
  \and
  \inferrule[\ruleAppRight]{t\in\struct
    \and u \rescbn{\rho}{\varphi}{\top} u'}{%
    t\app u \rescbn{\rho}{\varphi}{\mu} t\app u'}
  \and
  \inferrule[\ruleAbsTop]{%
    t \rescbn{\rho}{\varphi\cup\set{x}}{\top} t'}{%
    \l{x}.t \rescbn{\rho}{\varphi}{\top} \l{x}.t'}
  \and
  \inferrule[\ruleAbsBot]{%
    t \rescbn{\rho}{\varphi}{\bot} t'}{%
    \l{x}.t \rescbn{\rho}{\varphi}{\bot} \l{x}.t'}
  \and
  \inferrule[\ruleEsLeft]{%
    t \rescbn{\rho}{\varphi}{\mu} t'}{%
    t\esub{x}{u} \rescbn{\rho}{\varphi}{\mu} t'\esub{x}{u}}
  \and
  \inferrule[\ruleEsLeftFrozen]{%
    t \rescbn{\rho}{\varphi\cup\set{x}}{\mu} t' \and u\in\struct}{%
    t\esub{x}{u} \rescbn{\rho}{\varphi}{\mu} t'\esub{x}{u}}
  \and
  \inferrule[\ruleEsRight]{%
    t \rescbn{\ri{x}}{\varphi}{\mu} t%
    \and u \rescbn{\rho}{\varphi}{\bot} u'}{%
    t\esub{x}{u} \rescbn{\rho}{\varphi}{\mu} t\esub{x}{u'}}
  \and
  \inferrule[\ruleId]{~}{x \rescbn{\ri{x}}{\varphi}{\mu} x}
  \and
  \inferrule[\ruleSub]{~}{x \rescbn{\rs{x}{v}}{\varphi}{\mu} v}
  \and
  \inferrule[\ruleDb]{t \auxdb t'}{t \rescbn{\rdb}{\varphi}{\mu} t'}
  \and
  \inferrule[\ruleLsv]{%
    t \auxlsv{\varphi}{\mu} t'}{%
    t \rescbn{\rlsv}{\varphi}{\mu} t'}
  \end{mathpar}
  \caption{Reduction rules for \lscbn.}
  \label{fig:scbn-rules}
\end{figure}

Rules \ruleDb and \ruleLsv are the base cases for applying reductions \rdb or
\rlsv. As in Section~\ref{sec:weak-cbn}, each is defined using an auxiliary
reduction relation dealing with the list \lctx of explicit substitutions.
These auxiliary reductions are given in \autoref{fig:aux-rules}. They differ
from the ones from \autoref{fig:waux-rules} in two ways.

First, since the inference rule \ruleAuxLsv uses as a premise a reduction
$t\rescbn{\rs{x}{v}}{\varphi}{\mu}t'$, the auxiliary relation
\auxlsv{\varphi}{\mu} has to keep track of the context information. It is
thus parameterized by $\varphi$ and $\mu$.
The combination of these parameters and of the rules \ruleLsv and \ruleAuxLsv
makes it possible, in the case of a \rlsv-reduction, to resume the search for a
reducible variable in the context in which the substitution has been found
(instead of resetting the context). In~\cite{SCBN-ICFP}, a similar effect was
achieved using a more convoluted condition on a composition of contexts.

\begin{figure}[t]
  \begin{mathpar}
  \inferrule[\ruleAuxDb]{~}{(\l{x}.t)\app u \auxdb t\esub{x}{u}}
  \and
  \inferrule[\ruleAuxLsv]{%
    t \rescbn{\rs{x}{v}}{\varphi}{\mu} t'%
    \and v\text{ value}}{%
    t\esub{x}{v} \auxlsv{\varphi}{\mu} t'\esub{x}{v}}
  \\
  \inferrule[\ruleAuxDbSigma]{t\app u \auxdb z}{%
    t\esub{x}{w}\app u \auxdb z\esub{x}{w}}
  \and
  \inferrule[\ruleAuxLsvSigma]{%
    t\esub{x}{u} \auxlsv{\varphi}{\mu} t'}{%
    t\esub{x}{u\esub{y}{w}} \auxlsv{\varphi}{\mu} t'\esub{y}{w}}
  \and
  \inferrule[\ruleAuxLsvSigmaFrozen]{%
    t\esub{x}{u} \auxlsv{\varphi\cup\set{y}}{\mu} t'
    \and w\in\struct}{%
    t\esub{x}{u\esub{y}{w}} \auxlsv{\varphi}{\mu} t'\esub{y}{w}}
  \end{mathpar}
  \caption{Auxiliary reduction rules for \lscbn.}
  \label{fig:aux-rules}
\end{figure}

Second, a new rule \ruleAuxLsvSigmaFrozen appears, which strengthens the
already discussed \ruleAuxLsvSigma rule in the same way that \ruleEsLeftStruct
strengthens \ruleEsLeft.
The difference between \ruleAuxLsvSigma and \ruleAuxLsvSigmaFrozen
can be ignored until Section~\ref{sec:optimal-strategies}.

\paragraph*{Example.}
The reduction
$(\l{a}.a\app x)\esub{x}{(\l{y}.t)\esub{z}{u}\app v}
\red(\l{a}.a\app x)\esub{x}{t\esub{y}{v}\esub{z}{u}}$
is allowed by \lscbn, as shown by the following derivation.
The left branch of the derivation checks
that an occurrence of the variable $x$ is actually at a needed position in $\l{a}.a\app x$, while
its right branch reduces the argument of the substitution.
\begin{mathpar}
  \inferrule*[Right=\ruleEsRight]{%
    \inferrule*[Left=\ruleAbsTop]{%
      \inferrule*[Left=\ruleAppRight]{%
        a\in\struct[\set{a}]
        \and
        \inferrule*[Right=\ruleId]{~}{x\rescbn{\ri{x}}{\set{a}}{\top}x}
      }{%
        a\app x \rescbn{\ri{x}}{\set{a}}{\top} a\app x
      }
    }{%
      \l{a}.a\app x \rescbn{\ri{x}}{\emptyset}{\top} \l{a}.a\app x
    }
    \and
    \inferrule*[Right=\ruleDb]{%
      \inferrule*[Right=\ruleAuxDbSigma]{%
        \inferrule*[Right=\ruleAuxDb]{~}{(\l{y}.t)\app v \auxdb t\esub{y}{v}}
      }{%
        (\l{y}.t)\esub{z}{u}\app v
        \auxdb
        t\esub{y}{v}\esub{z}{u}
      }
    }{%
      (\l{y}.t)\esub{z}{u}\app v
      \rescbn{\rdb}{\emptyset}{\bot}
      t\esub{y}{v}\esub{z}{u}
    }
  }{%
    (\l{a}.a\app x)\esub{x}{(\l{y}.t)\esub{z}{u}\app v}
    \rescbn{\rdb}{\emptyset}{\top}
    (\l{a}.a\app x)\esub{x}{t\esub{y}{v}\esub{z}{u}}
  }
\end{mathpar}

This example also shows that top-level-like positions and evaluation positions
are different sets, even if their definitions are related. Indeed, we can
observe that the position actually reduced by the rule \ruleDb is labeled
by $\bot$. On the other hand, the definition of top-level-like positions in
\autoref{fig:frozen} may label as top-level-like a position that would not be
considered for reduction, as the position of $t$ in the term
$\l{x}.((\l{y}.y\;(x\;t))\;u)$, which is an argument position in a structure
(led by the frozen variable $x$) that is itself in argument position of another
application, the latter not being known to be a structure.

\paragraph{Basic properties}
The calculus \lscbn as we just defined it is not confluent. This can be observed
in the reduction of the term $(\l{x}.x)\app(\l{y}.(\l{z}.z)\app y)$ below.
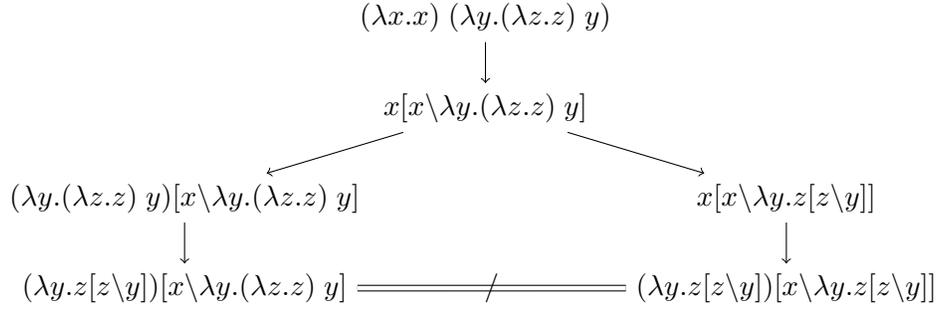
\begin{center}
  \begin{tikzpicture}[yscale=1.2]
    \node(base) {$(\l{x}.x)\app(\l{y}.(\l{z}.z)\app y)$};
    \node(split) at (0,-1) {$x\esub{x}{\l{y}.(\l{z}.z)\app y}$};
    \node(left1) at (-4,-2) {$(\l{y}.(\l{z}.z)\app y)\esub{x}{\l{y}.(\l{z}.z)\app y}$};
    \node(left2) at (-4,-3) {$(\l{y}.z\esub{z}{y})\esub{x}{\l{y}.(\l{z}.z)\app y}$};
    \node(right1) at (4, -2) {$x\esub{x}{\l{y}.z\esub{z}{y}}$};
    \node(right2) at (4, -3) {$(\l{y}.z\esub{z}{y})\esub{x}{\l{y}.z\esub{z}{y}}$};

    \draw[->] (base)--(split);
    \draw[->] (split)--(left1);
    \draw[->] (left1)--(left2);
    \draw[->] (split)--(right1);
    \draw[->] (right1)--(right2);
    \draw[double equal sign distance]  (left2)--(right2)
                              node[pos=0.5,sloped] {$/$};
  \end{tikzpicture}
\end{center}

After the first reduction step, we arrive at the term
$x\esub{x}{\l{y}.(\l{z}.z)\app y}$ where the only occurrence of $x$ is at
an evaluation position, and two reductions are possible: substituting the
term $\l{y}.(\l{z}.z)\app y$, which is a value, or reducing the redex
$(\l{z}.z)\app y$ inside this value. The latter case, on the right path,
leads to the normal form $(\l{y}.z\esub{z}{y})\esub{x}{\l{y}.z\esub{z}{y}}$
after reducing the only remaining redex.
The former case, on the left path, can only lead to
$(\l{y}.z\esub{z}{y})\esub{x}{\l{y}.(\l{z}.z)\app y}$, where the redex
$\l{y}.(\l{z}.z)\app y$ in the explicit substitution will never be accessed
again since the variable $x$ does not appear in the term anymore.
However, as a corollary of \autoref{thm:soundness}, we know that such a defect
of confluence may only happen inside an unreachable substitution. Moreover,
we will see in Section~\ref{sec:optimal-strategies} that strict confluence can
be restored in a restriction of the calculus.

Finally, note that the strong call-by-need strategy introduced
in~\cite{SCBN-ICFP} is included in our calculus.
One can recover this strategy by imposing two restrictions on
\rescbn{\rho}{\varphi}{\mu}:
\begin{itemize}
\item remove the rule \ruleAbsBot, so as to only reduce abstractions that are in top-level-like positions;
\item restrict the rule \ruleAppRight to the case where the left member
  of the application is a structure \emph{in normal form},
  since the strategy imposes left-to-right reduction of structures.
\end{itemize}

\subsection{Soundness}
\label{sec:soundness}

The calculus \lscbn is sound with respect to the \l-calculus, in the sense
that any normalizing reduction in \lscbn can be related to a normalizing
$\beta$-reduction through unfolding. This section establishes this result
(\autoref{thm:soundness}).
All the proofs in this section are formalized in Abella (see Section~\ref{sec:abella}).

The first part of the proof requires relating \lscbn-reduction to
$\beta$-reduction.
\begin{lem}[Simulation]\label{lem:simulation}
  If \textup{$t\red t'$} then $t\unfold \re^*_\beta {t'}\unfold$.
\end{lem}
\begin{proof}
  By induction on the reduction $t\rescbn{\rho}{\varphi}{\mu}t'$.
\end{proof}

The second part requires relating the normal forms of \lscbn to
$\beta$-normal forms.
The normal forms of \lscbn correspond to the normal forms of the strong
call-by-need strategy~\cite{SCBN-ICFP}. They can be characterized by the
inductive definition given in \autoref{fig:normalforms}.
Please note a subtlety of this definition: there are two rules for terms
of the form $t\esub{x}{u}$. One applies only when~$u$ is a structure,
and allows freezing the variable $x$ in $t$. The other applies without any
constraint on $u$, but requires $t$ to be a normal form without $x$ being
frozen (by our freshness convention on variables, $\varphi$ cannot contain
$x$). Since variables are considered normal forms only when frozen, the
latter case requires $x$ not to appear in any position of $t$ whose normality
is checked. In fact, for this rule to apply successfully, $x$ may appear only
in positions of $t$ that can be removed by one or more steps of \rgc.
Hence, $(\l{a}.a\app x)\esub{x}{r}$ is not a normal form if $r$ is a redex,
but $(\l{a}.a\app x)\esub{y}{z}\esub{z}{r}$ is, for any term $r$.
\begin{figure}[t]
  \begin{mathpar}
    \inferrule{x\in\varphi}{x\in\nf}
    \and
    \inferrule{t\in\nf \and t\in\struct \and u\in\nf}{t\app u\in\nf}
    \and
    \inferrule{t\in\nf[\varphi\cup\set{x}]}{\l{x}.t\in\nf}
    \and
    \inferrule{t\in\nf[\varphi\cup\set{x}] \and u\in\nf \and u\in\struct}{%
      t\esub{x}{u}\in\nf}
    \and
    \inferrule{t\in\nf}{t\esub{x}{u}\in\nf}
  \end{mathpar}
  \caption{Normal forms of \lscbn.}
  \label{fig:normalforms}
\end{figure}

\begin{lem}[Normal forms]\label{lem:normal-forms-lscbn}
  $t\in\nf$ if and only if there is no reduction
  \textup{$t\rescbn{\rho}{\varphi}{\mu}t'$}.
\end{lem}
\begin{proof}
  The first part (a term cannot be both in normal form and reducible) is by
  induction on the reduction rules. The second part (any term is
  either a normal form or a reducible term) is by induction on $t$.
\end{proof}
Please note that the characterization of normal forms in
\autoref{fig:normalforms} and the associated
\autoref{lem:normal-forms-lscbn} take into account the fact that we
did not include any \rgc rule in the calculus. If that were the case,
the characterization of normal forms could be adapted, by removing the
fifth rule.

\begin{lem}[Unfolded normal forms]\label{lem:normal-forms-lambda}
  If $t\in\nf$ then $t\unfold$ is a normal form in the \l-calculus.
\end{lem}
\begin{proof}
  By induction on $t\in\nf$.
\end{proof}

\noindent
Soundness is a direct consequence of the three previous lemmas.
\begin{thm}[Soundness]\label{thm:soundness}
  Let $t$ be a \textup{\lscbn}-term. If there is a reduction \textup{$t\red^* u$}
  with $u$ a \textup{\lscbn}-normal form, then $u\unfold$ is the $\beta$-normal
  form of $t\unfold$.
\end{thm}

This theorem implies that all the \lscbn-normal forms of a term $t$ are
equivalent modulo unfolding. This mitigates the fact that the calculus,
without a \rgc rule, is not confluent. For instance, the term
$(\l{x}.x)\app(\l{y}.(\l{z}.z)\app y)$ admits two normal forms
$(\l{y}.z\esub{z}{y})\esub{x}{\l{y}.(\l{z}.z)\app y}$ and
$(\l{y}.z\esub{z}{y})\esub{x}{\l{y}.z\esub{z}{y}}$, but both of them
unfold to $\l{y}.y$.

\subsection{Completeness}

Our strong call-by-need calculus is complete with respect to normalization in
the \l-calculus in a strong sense: Whenever a \l-term $t$ admits a normal form
in the pure \l-calculus, every reduction path in \lscbn eventually reaches a
representative of this normal form.
This section is devoted to proving this completeness result (\autoref{thm:completeness}).
The proof relies on the non-idempotent intersection type system in the
following way. First,
typability (\autoref{thm:typability}) ensures that any weakly normalizing
  \l-term admits a typing derivation (with no positive occurrence of \mset{}).
Second, we prove that any \lscbn-reduction in a typed \lscbn-term~$t$ (with
  no positive occurrence of \mset{}) is at a typed position of $t$
  (\autoref{thm:typed-reduction}).
Third, weighted subject reduction
  (\autoref{thm:weighted subject reduction}) provides a decreasing measure
  for \lscbn-reduction.
Finally, the obtained normal form is related to the $\beta$-normal form using
Lemmas~\ref{lem:simulation}, \ref{lem:normal-forms-lscbn},
and~\ref{lem:normal-forms-lambda}.

The proof of the forthcoming typed reduction
(\autoref{thm:typed-reduction}) uses a refinement of the
non-idempotent intersection types system of \lc,
given in \autoref{fig:annotated typing}.
Both systems derive the same typing judgments with the same typed
positions. The refined system, however, features an annotated typing
judgment $\smash{\ajdg{\Gamma}{\varphi}{\mu}{t}{\tau}}$ embedding the same
context information that are defined by the judgment $\ctx\vdash\mu,\varphi$
and used in the inductive reduction relation
\rescbn{\rho}{\varphi}{\mu}, namely the set $\varphi$ of frozen variables
at the considered position and a marker $\mu$ of top-level-like positions.
These annotations are faithful counterparts to the corresponding annotations
of \lscbn reduction rules; their information flows upward in the inference
rules following the same criteria.

In particular, the rule for typing an abstraction is split in two versions
\tyruleAbsBot and \tyruleAbsTop, the latter being applicable to $\top$
positions and thus freezing the variable bound by the abstraction
(in both rules, by freshness convention we assume $x\not\in\varphi$).
The rule for typing an application is also split into two versions:
\tyruleAppStruct is applicable when the left part of the application is a
structure and marks the right part as a $\top$ position, while
\tyruleApp is applicable otherwise. Note that this second rule allows the
argument of the application to be typed even if its position is not (yet)
reducible, but its typing is in a $\bot$ position.
Finally, the rule for typing an explicit substitution is similarly split into
two versions, depending on whether the content of the substitution is a
structure or not, and handling the set of frozen variables accordingly.
In both cases, the content of the substitution is typed in a $\bot$ position,
since this position is never top-level-like.

As it was the case for evaluation positions, the definition of typed positions
largely depends on the notions of frozen variables and top-level-like
positions, but the sets of typed positions and of top-level-like positions do
not have any particular relation to each other. Some typed positions are not
top-level-like (\emph{e.g.}, on the left of a typed application),
and some top-level-like positions are not typed (\emph{e.g.}, the argument of
an untyped structure). On the other hand, there is an actual connection between
the set of typed positions and the set of evaluation positions, given by
\autoref{thm:typed-reduction}.

We write \aderiv{\Gamma}{\varphi}{\mu}{t}{\tau} if there is a derivation $\Phi$
of the annotated typing judgment \ajdg{\Gamma}{\varphi}{\mu}{t}{\tau}.
We denote by $\fzt(\Phi)$ the set of types associated to
frozen variables in judgments of the derivation $\Phi$.

\begin{figure}[t]
\begin{mathpar}
  \inferrule[\tyruleVar]{~}{\ajdg{x:\mset{\sigma}}{\varphi}{\mu}{x}{\sigma}}
  \and
  \inferrule[\tyruleAbsBot]{\ajdg{\Gamma,x:\mvar}{\varphi}{\bot}{t}{\tau}}{%
    \ajdg{\Gamma}{\varphi}{\bot}{\l{x}.t}{\mvar\arr\tau}}
  \and
  \inferrule[\tyruleAbsTop]{\ajdg{\Gamma,x:\mvar}{\varphi\cup\set{x}}{\top}{t}{\tau}}{%
    \ajdg{\Gamma}{\varphi}{\top}{\l{x}.t}{\mvar\arr\tau}}
  \\
  \inferrule[\tyruleApp]{\ajdg{\Gamma}{\varphi}{\bot}{t}{\mvar\arr\tau}
    \and (\ajdg{\Delta_\sigma}{\varphi}{\bot}{u}{\sigma})_{\sigma\in\mvar}}{%
    \ajdg{\Gamma + \textstyle \sum_{\sigma\in\mvar} \Delta_\sigma}{\varphi}{\mu}{t\app u}{\tau}}
  \quad
  \inferrule[\tyruleAppStruct]{\ajdg{\Gamma}{\varphi}{\bot}{t}{\mvar\arr\tau}
    \quad t\in\struct
    \quad (\ajdg{\Delta_\sigma}{\varphi}{\top}{u}{\sigma})_{\sigma\in\mvar}}{%
    \ajdg{\Gamma + \textstyle \sum_{\sigma\in\mvar} \Delta_\sigma}{\varphi}{\mu}{t\app u}{\tau}}
  \\
  \inferrule[\tyruleEs]{\ajdg{\Gamma,x:\mvar}{\varphi}{\mu}{t}{\tau}
    \quad (\ajdg{\Delta_\sigma}{\varphi}{\bot}{u}{\sigma})_{\sigma\in\mvar}}{%
    \ajdg{\Gamma + \textstyle \sum_{\sigma\in\mvar} \Delta_\sigma}{\varphi}{\mu}{t\esub{x}{u}}{\tau}}
  \quad
  \inferrule[\tyruleEsStruct]{\ajdg{\Gamma,x:\mvar}{\varphi\cup\set{x}}{\mu}{t}{\tau}
    \quad u\in\struct
    \quad (\ajdg{\Delta_\sigma}{\varphi}{\bot}{u}{\sigma})_{\sigma\in\mvar}}{%
    \ajdg{\Gamma + \textstyle \sum_{\sigma\in\mvar} \Delta_\sigma}{\varphi}{\mu}{t\esub{x}{u}}{\tau}}
\end{mathpar}
  \caption{Annotated system of non-idempotent intersection types.}
  \label{fig:annotated typing}
\end{figure}

\begin{lem}[Typing derivation annotation]
  If there is a derivation \deriv{\Gamma}{t}{\tau},
  then for any $\varphi$ and $\mu$ there is a derivation
  \aderiv[\Phi']{\Gamma}{\varphi}{\mu}{t}{\tau}
  such that the sets of typed positions in $\Phi$ and $\Phi'$ are equal.
\end{lem}
\begin{proof}
  By induction on $\Phi$, since annotations do not interfere with typing.
\end{proof}
The converse property is also true, by erasure of the annotations, but it is not
used in the proof of the completeness result.

The most crucial part of the proof of \autoref{thm:typed-reduction} is ensuring that any argument of a typed
structure is itself at a typed position. This follows from the following three
lemmas.

\begin{lem}[Typed structure]\label{lem:typed-structure}
If \ajdg{\Gamma}{\varphi}{\mu}{t}{\tau} and $t\in\struct$,
    then there is $x\in\varphi$
    such that $\tau\in\tpos(\Gamma(x))$.
\end{lem}
\begin{proof}
  By induction on the structure of $t$. The most interesting case is
  the one of an explicit substitution $t_1\esub{x}{t_2}$. The induction
  hypothesis applied on $t_1$ can give the variable $x$ which does not appear
  in the conclusion, but in that case $t_2$ is guaranteed to be a structure
  whose type contains $\tau$. Let us give some more details.
\begin{itemize}
\item Case $t=x$. By inversion of $x\in\struct$ we deduce $x\in\varphi$.
  Moreover the only rule applicable to derive
  \ajdg{\Gamma}{\varphi}{\mu}{x}{\tau} is \tyruleVar, which gives the
  conclusion.
\item Case $t=t_1\app t_2$. By inversion of $t_1\app t_2\in\struct$ we deduce
  $t_1\in\struct$. Moreover the only rules applicable to
  derive \ajdg{\Gamma}{\varphi}{\mu}{t_1\app t_2}{\tau} are \tyruleApp and
  \tyruleAppStruct. Both have a premise
  \ajdg{\Gamma'}{\varphi}{\bot}{t_1}{\mvar\arr\tau}
  with $\Gamma'\subseteq\Gamma$, to which the induction hypothesis
  applies, ensuring $\mvar\arr\tau\in\tpos(\Gamma'(x))$ and thus
  $\tau\in\tpos(\Gamma'(x))$ and $\tau\in\tpos(\Gamma(x))$.
\item Case $t=t_1\esub{x}{t_2}$.
  We reason by case on the last rules applied to derive
  $t_1\esub{x}{t_2}\in\struct$ and
  \ajdg{\Gamma}{\varphi}{\mu}{t_1\esub{x}{t_2}}{\tau}.
  There are two possible rules for each.
  \begin{itemize}
  \item The case $t_1\esub{x}{t_2}\in\struct$ is deduced from
    $t_1\in\struct$ (with $x\not\in\varphi$),
    and \ajdg{\Gamma}{\varphi}{\mu}{t_1\esub{x}{t_2}}{\tau} comes
    from rule \tyruleEs. This rule has in particular a premise
    \ajdg{\Gamma'}{\varphi}{\mu}{t_1}{\tau} for
    a $\Gamma'=\Gamma'',x:\mvar$ such that $\Gamma''\subseteq\Gamma$.
    We thus have by induction hypothesis on $t_1$ that
    $\tau\in\tpos(\Gamma'(y))$ for some
    $y\in\varphi\cap\dom(\Gamma')$.
    Since $y\in\varphi$ and $x\not\in\varphi$, we have $y\neq x$.
    Then, $y\in\dom(\Gamma'')$, $y\in\dom(\Gamma)$,
    and $\Gamma(y)=\Gamma''(y)$.
  \item In the three other cases, we have:
    \begin{enumerate}
    \item a hypothesis $t_1\in\struct$ or $t_1\in\struct[\varphi\cup\set{x}]$,
      from which we deduce $t_1\in\struct[\varphi\cup\set{x}]$,
    \item a hypothesis \ajdg{\Gamma'}{\varphi}{\mu}{t_1}{\tau} or
      \ajdg{\Gamma'}{\varphi\cup\set{x}}{\mu}{t_1}{\tau} (for a
      $\Gamma'=\Gamma'',x:\mvar$ such that $\Gamma''\subseteq\Gamma$),
      from which we deduce \ajdg{\Gamma'}{\varphi\cup\set{x}}{\mu}{t_1}{\tau},
      and
    \item a hypothesis $t_2\in\struct$, coming from the
      derivation of $t_1\esub{x}{t_2}$ or the derivation of
      \ajdg{\Gamma}{\varphi}{\mu}{t_1\esub{x}{t_2}}{\tau} (or both).
    \end{enumerate}
    Then by induction hypothesis on $t_1$, we have
    $\tau\in\tpos(\Gamma'(y))$
    for some $y\in\varphi\cup\set{x}$.
    \begin{itemize}
    \item If $y\neq x$, then $y\in\varphi$ and $\Gamma(y)=\Gamma''(y)$,
      which allows a direct conclusion.
    \item If $y=x$, then $\tau\in\tpos(\Gamma'(x))$ implies $\mvar\neq\mset{}$.
      Let $\sigma\in\mvar$ with $\tau\in\tpos(\sigma)$.
      The instance of the rule \tyruleEs or \tyruleEsStruct we
      consider thus has at least one premise
      \ajdg{\Delta}{\varphi}{\bot}{t_2}{\sigma}
      with $\Delta\subseteq\Gamma$.
      Since $t_2\in\struct$, by induction hypothesis on $t_2$, there is
      $z\in\varphi\cap\dom(\Delta)$ such that
      $\sigma\in\tpos(\Delta(z))$.
      Then, $\tau\in\tpos(\Delta(z))$ and $\tau\in\Gamma$.
      \qedhere
    \end{itemize}
  \end{itemize}
\end{itemize}
\end{proof}

\begin{lem}[Subformula property]\label{lem:subformula}~
  \begin{enumerate}
  \item If \aderiv{\Gamma}{\varphi}{\top}{t}{\tau}
    then \(\left\{\begin{array}{rcl}
    \tpos(\fzt(\Phi)) & \subseteq &
    \bigcup_{x\in\varphi}\tpos(\Gamma(x))\cup\tneg(\tau) \\
    \tneg(\fzt(\Phi)) & \subseteq &
    \bigcup_{x\in\varphi}\tneg(\Gamma(x))\cup\tpos(\tau)
    \end{array}\right.\)
  \item If \aderiv{\Gamma}{\varphi}{\bot}{t}{\tau}
    then \(\left\{\begin{array}{rcl}
    \tpos(\fzt(\Phi)) & \subseteq & \bigcup_{x\in\varphi}\tpos(\Gamma(x)) \\
    \tneg(\fzt(\Phi)) & \subseteq & \bigcup_{x\in\varphi}\tneg(\Gamma(x))
    \end{array}\right.\)
  \end{enumerate}
\end{lem}
\begin{proof}
  By mutual induction on the typing derivations.
  Most cases are fairly straightforward. The only difficult case comes
  from the rule \tyruleAppStruct, in which there is a premise
  $\ajdg{\Delta}{\varphi}{\top}{u}{\sigma}$ with mode $\top$ but with a type
  $\sigma$ that does not clearly appear in the conclusion. Here we need
  the typed structure (\autoref{lem:typed-structure}) to conclude.
  Let us give some more details.
\begin{itemize}
\item Both properties are immediate in case \tyruleVar, where
  $\fzt(\Phi)=\set{\sigma}$.
\item Cases for abstractions.
  \begin{itemize}
  \item If \aderiv{\Gamma}{\varphi}{\bot}{\l{x}.t}{\mvar\arr\tau}
    by rule \tyruleAbsBot with premise
    \aderiv[\Phi']{\Gamma,x:\mvar}{\varphi}{\bot}{t}{\tau}.
    Write $\Gamma'=\Gamma,x:\mvar$.
    By induction hypothesis we have
    $\tpos(\fzt(\Phi'))\subseteq\bigcup_{y\in\varphi}\tpos(\Gamma'(y))$.
    Since $x\not\in\varphi$ by renaming convention, we deduce that
    $\tpos(\fzt(\Phi'))\subseteq\bigcup_{y\in\varphi}\tpos(\Gamma(y))$
    and
    $\tpos(\fzt(\Phi))\subseteq\bigcup_{y\in\varphi}\tpos(\Gamma(y))$.
    The same applies to negative type occurrences, which concludes the case.
  \item If \aderiv{\Gamma}{\varphi}{\top}{\l{x}.t}{\mvar\arr\tau}
    by rule \tyruleAbsTop with premise
    \aderiv[\Phi']{\Gamma,x:\mvar}{\varphi\cup\set{x}}{\top}{t}{\tau}.
    Write $\Gamma'=\Gamma,x:\mvar$.
    By induction hypothesis we have
    \[\begin{array}{rcl}
    \tpos(\fzt(\Phi')) & \subseteq &
    \bigcup_{y\in(\varphi\cup\set{x})}\tpos(\Gamma'(y))\cup\tneg(\tau) \\
    & = & \bigcup_{y\in\varphi}\tpos(\Gamma(y))\cup\tpos(\mvar)\cup\tneg(\tau) \\
    & = & \bigcup_{y\in\varphi}\tpos(\Gamma(y))\cup\tneg(\mvar\arr\tau)
    \end{array}\]
    Thus $\tpos(\fzt(\Phi))\subseteq\bigcup_{y\in\varphi}\tpos(\Gamma(y))\cup\tneg(\mvar\arr\tau)$.
    The same applies to negative occurrences, which concludes the case.
  \end{itemize}
\item Cases for application.
  \begin{itemize}
  \item Cases for \tyruleApp are by immediate application of the induction
    hypotheses.
  \item If \aderiv{\Gamma}{\varphi}{\mu}{t\app u}{\tau}
    by rule \tyruleAppStruct, with premises
    \aderiv[\Phi_t]{\Gamma_t}{\varphi}{\bot}{t}{\mvar\arr\tau},
    $t\in\struct$ and
    \aderiv[\Phi_\sigma]{\Delta_\sigma}{\varphi}{\top}{u}{\sigma}
    for $\sigma\in\mvar$, with $\Gamma_t\subseteq\Gamma$ and
    $\Gamma_\sigma\subseteq\Gamma$ for all $\sigma\in\mvar$.
    Independently of the value of $\mu$, we show that
    $\tpos(\fzt(\Phi))\subseteq\bigcup_{x\in\varphi}\tpos(\Gamma(x))$
    and $\tneg(\fzt(\Phi))\subseteq\bigcup_{x\in\varphi}\tneg(\Gamma(x))$
    to conclude on both sides of the mutual induction.

    Directly from the induction hypothesis,
    \(\tpos(\fzt(\Phi_t))\subseteq\bigcup_{x\in\varphi}\Gamma_t(x)
    \subseteq\tpos(\fzt(\Phi))\).
    By induction hypothesis on the other premises we have
    $\tpos(\fzt(\Phi_\sigma))\subseteq\bigcup_{x\in\varphi}\Gamma_\sigma(x)\cup\tneg(\tau)$
    for $\sigma\in\mvar$.
    We immediately have
    $\bigcup_{x\in\varphi}\Gamma_\sigma(x)\subseteq\bigcup_{x\in\varphi}\Gamma(x)$.
    We conclude by showing that $\tneg(\sigma)\subseteq\tpos(\Gamma_t(x))$ for
    some $x\in\varphi$.
    Since $t\in\struct$, by the first subformula property and the typing
    hypothesis on $t$ we deduce that
    there is a $x\in\varphi$ such that $\mvar\arr\tau\in\tpos(\Gamma_t(x))$.
    By closeness of type occurrences sets $\tpos(\tau)$ this means
    $\tpos(\mvar\arr\tau)\subseteq\tpos(\Gamma_t(x))$.
    By definition
    $\tpos(\mvar\arr\tau)=\tneg(\mvar)\cup\tpos(\tau)=\bigcup_{\sigma\in\mvar}\tneg(\sigma)\cup\tpos(\tau)$, which allows us to conclude the proof that
    $\bigcup_{x\in\varphi}\Gamma_\sigma(x)\cup\tneg(\tau) \subseteq \bigcup_{x\in\varphi}\Gamma(x)$.

    The same argument also applies to negative positions, and concludes the case.
  \end{itemize}
\item Cases for explicit substitution immediately follow the induction hypothesis.
  \qedhere
\end{itemize}
\end{proof}

\begin{lem}[Typed structure argument]\label{lem:typed-structure-argument}
  If \aderiv{\Gamma}{\varphi}{\mu}{t}{\tau}
  with $\mset{}\not\in\tpos(\jdg{\Gamma}{t}{\tau})$,
  then every typing judgment of the shape
  \ajdg{\Gamma'}{\varphi'}{\mu'}{s}{\mvar\arr\sigma}
  in $\Phi$ with $s\in\struct[\varphi']$ satisfies
  $\mvar\neq\mset{}$.
\end{lem}
\begin{proof}
  Let \ajdg{\Gamma'}{\varphi'}{\mu'}{s}{\mvar\arr\sigma} in $\Phi$
  with $s\in\struct[\varphi']$.
  By \autoref{lem:typed-structure}, there is $x\in\varphi'$
  such that $\mvar\arr\sigma\in\tpos(\Gamma'(x))$.
  Then $\mvar\in\tneg(\Gamma'(x))$ and $\mvar\in\tneg(\fzt(\Phi))$.
  By \autoref{lem:subformula},
  $\mvar\in\tpos(\ajdg{\Gamma}{\varphi}{\mu}{t}{\tau})$,
  thus $\mvar\neq\mset{}$.
\end{proof}

\begin{thm}[Typed reduction]\label{thm:typed-reduction}
  If \aderiv{\Gamma}{\varphi}{\mu}{t}{\tau}
  with $\mset{}\not\in\tpos(\jdg{\Gamma}{t}{\tau})$,
  then every \textup{\lscbn}-reduction \textup{$t\rescbn{\rho}{\varphi}{\mu}t'$}
  is at a typed position.
\end{thm}
\begin{proof}
  We prove by induction on $t\rescbn{\rho}{\varphi}{\mu}t'$
  that, if $\smash{\aderiv{\Gamma}{\varphi}{\mu}{t}{\tau}}$ with $\Phi$ such that
  any typing judgment \ajdg{\Gamma'}{\varphi'}{\mu'}{s}{\mvar\arr\sigma}
  in $\Phi$ with $s\in\struct[\varphi']$ satisfies $\mvar\neq\mset{}$,
  then $t\rescbn{\rho}{\varphi}{\mu}t'$ reduces at a typed position
  (the restriction on $\Phi$ is enabled by
  \autoref{lem:typed-structure-argument}).
  Since all the other reduction cases concern positions that are systematically
  typed, we focus here on \ruleAppRight and \ruleEsRight.
  \begin{itemize}
  \item Case \ruleAppRight: $t\app u\rescbn{\rho}{\varphi}{\mu} t\app u'$
    with $t\in\struct$ and $u\rescbn{\rho}{\varphi}{\top}u'$, assuming
    \aderiv[\Phi]{\Gamma}{\varphi}{\mu}{t\app u}{\sigma}.
    By inversion of the last rule in $\Phi$ we know there is a subderivation
    \aderiv[\Phi']{\Gamma'}{\varphi}{\bot}{t}{\mvar\arr\sigma}
    and by hypothesis $\mvar\neq\mset{}$.
    Then $u$ is typed in $\Phi$ and we can conclude by induction hypothesis.
  \item Case \ruleEsRight: $t\esub{x}{u}\rescbn{\rho}{\varphi}{\mu}t\esub{x}{u'}$
    with $t\rescbn{\ri{x}}{\varphi}{\mu}t'$
    and $u\rescbn{\rho}{\varphi}{\bot}u'$, assuming
    $\smash{\aderiv[\Phi]{\Gamma}{\varphi}{\mu}{t\esub{x}{u}}{\tau}}$.
    By inversion of the last rule in $\Phi$ we kown there is a subderivation
    \aderiv[\Phi']{\Gamma',x:\mvar}{\varphi}{\mu}{t}{\tau}.
    By induction hypothesis we know that reduction
    $t\rescbn{\smash{\ri{x}}}{\varphi}{\mu}t'$ is at a typed position in $\Phi'$,
    thus $x$ is typed in $t$ and $\mvar\neq\mset{}$.
    Then $u$ is typed in $\Phi$ and we can conclude by induction
    hypothesis on $u\rescbn{\rho}{\varphi}{\bot}u'$.
    \qedhere
  \end{itemize}
\end{proof}

\begin{thm}[Completeness]\label{thm:completeness}
  If a \l-term $t$ is weakly normalizing in the \l-calculus,
  then $t$ is strongly normalizing in \textup{\lscbn}.
  Moreover, if $n_\beta$ is the normal form of $t$ in the \l-calculus,
  then any normal form \textup{$n_\rscbn$} of $t$ in \textup{\lscbn} is such that
  \textup{$n_\rscbn\unfold = n_\beta$}.
\end{thm}
\begin{proof}
  Let $t$ be a pure \l-term that admits a normal form $n_\beta$ for
  $\beta$-reduction. By \autoref{thm:typability} there exists
  a derivable typing judgment \jdg{\Gamma}{t}{\tau} such that
  $\mset{}\not\in\tpos(\jdg{\Gamma}{t}{\tau})$.
  Thus by Theorems~\ref{thm:typed-reduction}
  and~\ref{thm:weighted subject reduction}, the term $t$ is
  strongly normalizing for \red.
  Let $t\re_\rscbn^*n_\rscbn$ be a maximal reduction in \lscbn.
  By \autoref{lem:normal-forms-lscbn}, $n_\rscbn\in\nf$, and by
  \autoref{lem:normal-forms-lambda}, $n_\rscbn\unfold$ is a normal form
  in the \l-calculus.
  Moreover, by simulation (\autoref{lem:simulation}), there is a reduction
  $t\unfold\reb^*n_\rscbn\unfold$. By uniqueness of the normal form in the
  \l-calculus, $n_\rscbn\unfold=n_\beta$.
\end{proof}

Note that, despite the fact that \lscbn does not enjoy the diamond property,
our theorems of soundness (\autoref{thm:soundness}) and completeness
(\autoref{thm:completeness}) imply that, in \lscbn, a term is weakly normalizing
if and only if it is strongly normalizing.

These termination and completeness results also allow us to formalize our
claim that \lscbn reduces only ``needed'' redexes.

\begin{thm}[Neededness]\label{thm:neededness}
  Consider a \textup{\lscbn}-term $t$ such that $t\unfold$ is weakly
  normalizing in the \l-calculus.
  For any \textup{\rdb}-reduction \textup{$t\rescbn{\rdb}{\varphi}{\mu}t'$}
  in \textup{\lscbn}, the subterm $r$ of $t$ that is reduced is such that at
  least one occurrence of the corresponding $\beta$-redex $r'$ in $t\unfold$
  is needed.
\end{thm}
\begin{proof}
  Suppose no occurrence of the redex $r'$ is needed in $t\unfold$.
  Then by definition there is a $\beta$-reduction sequence of $t\unfold$
  that leads to the normal form
  of $t\unfold$ without reducing any copy of~$r'$ nor any of its residuals (and
  the normal form, being normal, does not contain any residual of~$r'$).
  Write $p$ the position of $r$ in $t$, and $t(\Omega)_p$ the term obtained
  from $t$ by replacing~$r$ by the diverging term $\Omega$.
  Then the term $(t(\Omega)_p)\unfold$ is equal to the term obtained from
  $t\unfold$ by replacing each copy of $r'$ by $\Omega$, and can still be
  $\beta$-normalized by the same $\beta$-reduction sequence. Thus, this term is
  weakly normalizing in the \l-calculus.

  Since $(t(\Omega)_p)\revert$ $\beta$-reduces to $(t(\Omega)_p)\unfold$,
  we deduce that the pure term $(t(\Omega)_p)\revert$ is also weakly
  normalizing in the \l-calculus. Thus by \autoref{thm:typability}
  there is a typing judgment \jdg{\Gamma}{(t(\Omega)_p)\revert}{\tau}
  such that $\mset{}\not\in\tpos(\jdg{\Gamma}{t}{\tau})$, and by
  \autoref{prop:revert typing} we also have \jdg{\Gamma}{t(\Omega)_p}{\tau}.
  Then by Theorems~\ref{thm:typed-reduction}
  and~\ref{thm:weighted subject reduction}, the term $t(\Omega)_{p}$ is
  strongly normalizing in \lscbn.

  However, if $\re_\rscbn$ allows reducing $r$ in $t$, then it also allows
  reducing $\Omega$ in $t(\Omega)_{p}$, as well as any reduct of $\Omega$
  at the same position, which implies the existence of an infinite
  $\re_\rscbn$-reduction sequence from $t$: contradiction.  
  Thus at least one occurrence of the $\beta$-redex $r'$ is needed
  in~$t\unfold$.
\end{proof}

\section{Relatively optimal strategies}
\label{sec:optimal-strategies}

Our proposed \lscbn-calculus guarantees that, in the process of reducing
a term to its strong normal form, only needed redexes are ever reduced.
This does not tell anything about the length of reduction sequences,
though. Indeed, a term might be substituted several times before being
reduced, thus leading to duplicate computations. To prevent this
duplication, we introduce a notion of \emph{local normal form}, which
is used to restrict the \emph{value} criterion in the \ruleAuxLsv
rule.
This restricted calculus, named \lscbnp, has the same rules as \lscbn
(Figures~\ref{fig:scbn-rules} and~\ref{fig:aux-rules}),
except that \ruleAuxLsv is replaced by the rule shown in
\autoref{fig:aux-rules-plus}.

\begin{figure}[t]
  \begin{mathpar}
    \inferrule[\ruleAuxLsv]{%
      t \rescbn{\rs{x}{v}}{\varphi}{\mu} t'%
      \and v\in\lnf[\varphi,\emptyset,\bot]%
      \and v\text{ value}
    }{%
    t\esub{x}{v} \auxlsv{\varphi}{\mu} t'\esub{x}{v}}
  \end{mathpar}

  \caption{New rule \ruleAuxLsv for \lscbnp.}
  \label{fig:aux-rules-plus}
\end{figure}

We then show that this restriction is
strong enough to guarantee the diamond property. Finally, we explain why
our restricted calculus only produces minimal sequences,
among all the reduction sequences allowed by \lscbn. This makes it
a relatively optimal strategy.

\subsection{Local normal forms}

In \lc and \lscbn, substituted terms can be arbitrary values. In
particular, they might be abstractions whose body contains some redexes. Since
substituted variables can appear multiple times, this would cause the
redex to be reduced several times if the value is substituted too soon.
Let us illustrate this phenomenon on the following example, where $\textsf{id} =
\l{x}.x$. The sequence of reductions does not depend on the set $\varphi$
of frozen variables nor on the position~$\mu$, so we do not write them to
lighten a bit the notations. Subterms that are about to be substituted or
reduced are underlined.

\[\begin{array}{rcl}
\underline{(\l{w}.w\app w)\app(\l{y}.\textsf{id}\app y)}
&\literescbn{db}&(\underline{w}\app w)\esub{w}{\l{y}.\textsf{id}\app y}\\
&\literescbn{lsv}&((\l{y}.\underline{\textsf{id}\app y})\app w)\esub{w}{\l{y}.\textsf{id}\app y}\\
&\literescbn{db}&\underline{((\l{y}.x\esub{x}{y})\app w)}\esub{w}{\l{y}.\textsf{id}\app y}\\
&\literescbn{db}&\underline{x}\esub{x}{y}\esub{y}{w}\esub{w}{\l{y}.\textsf{id}\app y}\\
&\literescbn{lsv$\null\times 3$}&(\l{y}.\underline{\textsf{id}\app y})\esub{x}{\l{y}.\textsf{id}\app y}\esub{y}{\l{y}.\textsf{id}\app y}\esub{w}{\l{y}.\textsf{id}\app y}\\
&\literescbn{db}&(\l{y}.x\esub{x}{y})\esub{x}{\l{y}.\textsf{id}\app y}\esub{y}{\l{y}.\textsf{id}\app y}\esub{w}{\l{y}.\textsf{id}\app y}\\
\end{array}\]
Notice how $\textsf{id}\app y$ is reduced twice, which would not have
happened if the second reduction had focused on the body of the
abstraction.

This suggests that a substitution should only be allowed if the
substituted term is in normal form. But such a strong requirement is
incompatible with our calculus, as it would prevent the abstraction
$\l{y}.y\app\Omega$ (with $\Omega$ a diverging term) to ever be
substituted in the following example, thus preventing normalization (with $a$ a closed term).
\[\begin{array}{rcl}
\underline{w}\app (\l{x}.a)\esub{w}{\l{y}.y\app\Omega}
&\literescbn{lsv}&\underline{(\l{y}.y\app\Omega)\app (\l{x}.a)}\esub{w}{\l{y}.y\app\Omega}\\
&\literescbn{db}&(\underline{y}\app \Omega)\esub{y}{\l{x}.a}\esub{w}{\l{y}.y\app\Omega}\\
&\literescbn{lsv}&(\underline{(\l{x}.a)\app \Omega})\esub{y}{\l{x}.a}\esub{w}{\l{y}.y\app\Omega}\\
&\literescbn{db}&a\esub{x}{\Omega}\esub{y}{\l{x}.a}\esub{w}{\l{y}.y\app\Omega}\\
\end{array}\]
Notice how the sequence of reductions has progressively removed all the
occurrences of $\Omega$, until the only term left to reduce is the closed
term $a$.

To summarize, substituting any value is too permissive and can cause
duplicate computations, while substituting only normal forms is too
restrictive as it prevents normalization. So, we need some relaxed notion
of normal form, which we call \emph{local normal form}. The intuition is
as follows. The term $\l{y}.y\app\Omega$ is not in normal form, because
it could be reduced if it was in a $\top$ position. But in a $\bot$
position, variable $y$ is not frozen, which prevents any further
reduction of $y\app\Omega$. The inference rules are
presented in \autoref{fig:local-nf}.

\begin{figure}[t]
  \begin{mathpar}
    \inferrule[\textsc{var}]{x\in\varphi\cup\omega}{x\in\Vewt}
    \and
    \inferrule[\textsc{abs-$\top$}]{t\in \Vexwtt}{\l{x}.t\in \Vewtt}
    \and
    \inferrule[\textsc{abs-$\bot$}]{t\in \Vewxtb} {\l{x}.t\in\Vewtb}
    \and
    \inferrule[\textsc{es}]{t\in \Vewt}{t[x\setminus u]\in \Vewt}
    \\
    \inferrule[\textsc{app-$\varphi$}]{t\in \Vewt \and t\in \struct
\and u\in\Vewtt}
    {t~u\in\Vewt}
    \and
    \inferrule[\textsc{app-$\omega$}]{t\in \Vewt \and t\in
      \struct[\omega] }{t~u\in\Vewt}
    \\
    \inferrule[\textsc{es-$\varphi$}]{t\in \Vexwt \and u\in \Vewtb \and
    u\in \struct}{t[x\setminus u] \in \Vewt}
    \and
    \inferrule[\textsc{es-$\omega$}]{t\in \Vewxt \and u\in \struct[\omega]}
    {t[x\setminus u] \in \Vewt}
  \end{mathpar}
  \caption{Local normal forms.}
  \label{fig:local-nf}
\end{figure}

If an abstraction is in a $\top$ position, its variable is added to the
set $\varphi$ of frozen variables, as in \autoref{fig:scbn-rules}. But
if an abstraction is in a $\bot$ position, its variable is added to a new
set $\omega$, as shown in rule \textsc{abs-$\bot$} of
\autoref{fig:local-nf}. That is what will happen to $y$ in
$\l{y}.y\app\Omega$.

For an application, the left part is still required to be
a structure. But if the leading variable of the structure is not frozen
(and thus in $\omega$), our \lscbn-calculus guarantees that no reduction
will occur in the right part of the application. So, this part
does not need to be constrained in any way. This is rule
\textsc{app-$\omega$} of \autoref{fig:local-nf}. It applies to our
example, since~$y\app\Omega$ is a structure led by $y\in\omega$.
Substitutions are handled in a similar way, as shown by rule~\textsc{es-$\omega$}.

\paragraph*{Example.}
The justification that our example argument $\l{y}.y\app\Omega$ is a local
normal form in a \emph{non-}top-level-like position is summed up by the
following derivation.
\begin{mathpar}
  \inferrule*[Right=\textsc{abs-$\bot$}]{%
    \inferrule*[Right=\textsc{app-$\omega$}]{%
      \inferrule*[Left=\textsc{var}]{%
        y\in\set{y}
      }{
        y\in\lnf[\emptyset,\set{y},\bot]
      }
      \and
      y\in\struct[\omega]
    }{%
      y\app\Omega\in\lnf[\emptyset,\set{y},\bot]
    }
  }{%
    \l{y}.y\app\Omega\in\lnf[\emptyset,\emptyset,\bot]
  }
\end{mathpar}

The definition of local normal forms (\autoref{fig:local-nf}) is consistent
with the already given definition of normal forms (\autoref{fig:normalforms}).
In a sense, normal forms are top-level local normal forms.

\begin{lem}[Local normal forms]\label{lem:local-normal-form}
  $t\in\nf$ if and only if $t\in\lnf[\varphi,\emptyset,\top]$.
\end{lem}
\begin{proof}
  Each direction is by induction on the derivation of (local) normality.
  Dealing with the rule \textsc{es-$\varphi$} requires a lemma stating that
  when $u\in\struct$, we have $u\in\Vewtt$ if and only if $u\in\Vewtb$.
\end{proof}

This implies that the two calculi \lscbn and \lscbnp have the same normal
forms, and that the completeness theorem stated for
\lscbn (\autoref{thm:completeness}) is still valid for the restricted calculus~\lscbnp.

\subsection{Diamond property}

As mentioned before, in both \lc and \lscbn, terms might be substituted
as soon as they are values, thus potentially causing duplicate
computations, and even breaking confluence.
Things are different however in our restricted calculus \lscbnp.

Consider for instance the counterexample to confluence given at the
end of Section~\ref{sec:red-lscbn}. Given the term
$x\esub{x}{\l{y}.(\l{z}.z)\app y}$, the calculus \lscbn allows
both substituting the value $\l{y}.(\l{z}.z)\app y$ or reducing it
further. The restricted calculus \lscbnp, however, prevents the
substitution of $\l{y}.(\l{z}.z)\app y$, since this value is not a
local normal form. Hence, \lscbnp allows only one of the two paths
that are possible in \lscbn, which is enough to restore confluence.
Actually, \lscbnp even enjoys the stronger \emph{diamond property},
that is, confluence in one step.

\begin{thm}[Diamond]\label{thm:diamond}
Suppose \textup{$t\rescbnp{\rho_1}{\varphi}{\mu}t_1$} and
\textup{$t\rescbnp{\rho_2}{\varphi}{\mu}t_2$} with $t_1\neq t_2$.
Assume that, if $\rho_1$ and
$\rho_2$ are \textup{\textsf{sub}} or \textup{\textsf{id}}, then they apply
to separate variables and that, if $\rho_1$ or $\rho_2$ is
\textup{\textsf{sub}}, then it applies to a variable that is not in $\varphi$.
Then there exists $t'$ such that
\textup{$t_1\rescbnp{\rho_2}{\varphi}{\mu}t'$} and
\textup{$t_2\rescbnp{\rho_1}{\varphi}{\mu}t'$}.
\end{thm}

\begin{proof}
  The statement of the diamond theorem has first to be
  generalized so that the steps $t\to t_1$ and $t\to t_2$ can use the main
  reduction \rescbn{\rho}{\varphi}{\mu} or
  the auxiliary reductions \auxdb and \auxlsv{\varphi}{\mu}.
  The proof is then an induction on the size of $t$, with tedious reasoning
  by case on the shape of $t$ and on the last inference rule applied on each
  side. Most cases are rather unsurprising.
  We present here selected subcases applying to the shape $t\esub{x}{u}$,
  which illustrate the main subtleties encountered in the whole proof.
  \begin{itemize}
  \item Case \ruleEsLeft vs \ruleEsRight. More precisely we have
    \begin{itemize}
    \item $t\esub{x}{u}\rescbnp{\rho_1}{\varphi}{\mu}t_1\esub{x}{u}$
      by rule \ruleEsLeft with $t\rescbnp{\rho_1}{\varphi}{\mu}t_1$,
      and
    \item $t\esub{x}{u}\rescbnp{\rho_2}{\varphi}{\mu}t\esub{x}{u_2}$
      by rule \ruleEsRight with $t\rescbnp{\ri{x}}{\varphi}{\mu}t$
      and $u\rescbnp{\rho_2}{\varphi}{\mu}u_2$.
    \end{itemize}
    Notice that, by freshness of $x$, the rule $\rho_1$ cannot be a
    \textsf{sub} or \textsf{id} applying to the variable $x$. Then we
    can apply our induction hypothesis to
    $t\rescbnp{\rho_1}{\varphi}{\mu}t_1$ and
    $t\rescbnp{\ri{x}}{\varphi}{\mu}t$ to deduce that
    $t_1\rescbnp{\ri{x}}{\varphi}{\mu}t_1$.
    With rule \ruleEsRight we then deduce
    $t_1\esub{x}{u}\rescbnp{\rho_2}{\varphi}{\mu}t_1\esub{x}{u_2}$.
    Moreover, the rule \ruleEsLeft applied to
    $t\rescbnp{\rho_1}{\varphi}{\mu}t_1$ immediately gives
    $t\esub{x}{u_2}\rescbnp{\rho_2}{\varphi}{\mu}t_1\esub{x}{u_2}$
    and closes the case.
  \item Case \ruleEsLeftFrozen vs \ruleEsRight. This is superficially similar to the previous one,
    but adds some subtleties, as we now have
    \begin{itemize}
    \item $t\esub{x}{u}\rescbnp{\rho_1}{\varphi}{\mu}t_1\esub{x}{u}$
      by rule \ruleEsLeftFrozen with $u\in\struct$ and
      $t\rescbnp{\rho_1}{\varphi\cup\set{x}}{\mu}t_1$, and
    \item $t\esub{x}{u}\rescbnp{\rho_2}{\varphi}{\mu}t\esub{x}{u_2}$
      by rule \ruleEsRight with $t\rescbnp{\ri{x}}{\varphi}{\mu}t$
      and $u\rescbnp{\rho_2}{\varphi}{\mu}u_2$.
    \end{itemize}
    The reduction $t\rescbnp{\rho_1}{\varphi\cup\set{x}}{\mu}t_1$ uses the
    frozen variables $\varphi\cup\set{x}$ rather than just $\varphi$.
    To apply the induction hypothesis, we first have to weaken
    $t\rescbnp{\ri{x}}{\varphi}{\mu}t$ into
    $t\rescbnp{\ri{x}}{\varphi\cup\set{x}}{\mu}t$
    (\autoref{lem:phi-weakening} below).
    The induction hypothesis then gives
    $t_1\rescbnp{\ri{x}}{\varphi\cup\set{x}}{\mu}t_1$, which then has to
    be strengthened into $t_1\rescbnp{\ri{x}}{\varphi}{\mu}t_1$
    (\autoref{lem:phi-id-strengthening} below)
    to obtain $t_1\esub{x}{u}\rescbnp{\rho_2}{\varphi}{\mu}t_1\esub{x}{u_2}$
    with rule \ruleEsRight.

    Finally, to close the case with rule \ruleEsLeftFrozen we also need
    to justify that $u_2\in\struct$, which we do using a stability property
    of structures (\autoref{lem:stab-struct} below).
  \item Case \ruleEsLeftFrozen vs \ruleLsv.
    These cases are not compatible. Indeed \ruleEsLeftFrozen applies to
    a term $t\esub{x}{u}$ where $u\in\struct$, whereas \ruleLsv and the
    associated auxiliary rules apply to a term $t\esub{x}{u}$ where $u$
    is an answer.
  \item Case \ruleEsLeft vs \ruleAuxLsv. In this case our starting term
    has necessarily the shape $t\esub{x}{v}$ with $v$ a value and we have
    \begin{itemize}
    \item $t\esub{x}{v}\rescbnp{\rho_1}{\varphi}{\mu}t_1\esub{x}{v}$
      by rule \ruleEsLeft with $t\rescbnp{\rho_1}{\varphi}{\mu}t_1$, and
    \item $t\esub{x}{v}\auxlsv{\varphi}{\mu}t_2\esub{x}{v}$
      by rule \ruleAuxLsv with $t\rescbn{\rs{x}{v}}{\varphi}{\mu}t_2$.
    \end{itemize}
    Due to freshness, the variable $x$ cannot appear in $\rho_1$ nor
    in $\varphi$. Then by induction hypothesis we obtain $t_3$ such that
    $t_1\rescbn{\rs{x}{v}}{\varphi}{\mu}t_3$ and
    $t_2\rescbn{\rho_1}{\varphi}{\mu}t_3$. By applying rule \ruleAuxLsv to
    the first reduction we obtain
    $t_1\esub{x}{v}\auxlsv{\varphi}{\mu}t_3\esub{x}{v}$,
    and by applying rule \ruleEsLeft to the second reduction we obtain
    $t_2\esub{x}{v}\rescbn{\rho_1}{\varphi}{\mu}t_3\esub{x}{v}$ and
    conclude.
  \item Case \ruleEsLeft vs \ruleAuxLsvSigma. In this case our starting term
    has necessarily the shape $t\esub{x}{u\esub{y}{w}}$ and we have
    \begin{itemize}
    \item $t\esub{x}{u\esub{y}{w}}\rescbnp{\rho_1}{\varphi}{\mu}t_1\esub{x}{u\esub{y}{w}}$
      by rule \ruleEsLeft with $t\rescbnp{\rho_1}{\varphi}{\mu}t_1$, and
    \item $t\esub{x}{u\esub{y}{w}}\auxlsv{\varphi}{\mu}t_2\esub{y}{w}$
      by rule \ruleAuxLsvSigma with $t\esub{x}{u}\auxlsv{\varphi}{\mu}t_2$.
    \end{itemize}
    From the premise $t\rescbnp{\rho_1}{\varphi}{\mu}t_1$, by rule
    \ruleEsLeft we deduce
    $t\esub{x}{u}\rescbnp{\rho_1}{\varphi}{\mu}t_1\esub{x}{u}$.
    Since the term $t\esub{x}{u}$ is smaller than $t\esub{x}{u\esub{y}{w}}$,
    we apply the induction hypothesis to obtain a term $t_3$ such that
    $t_1\esub{x}{u}\auxlsv{\varphi}{\mu}t_3$ and
    $t_2\rescbnp{\rho_1}{\varphi}{\mu}t_3$.

    From the first reduction, by rule \ruleAuxLsvSigma we deduce
    $t_1\esub{x}{u\esub{y}{w}}\auxlsv{\varphi}{\mu}t_3\esub{y}{w}$
    and from the second reduction, by rule \ruleEsLeft we deduce
    $t_2\esub{y}{w}\rescbnp{\rho_1}{\varphi}{\mu}t_3\esub{y}{w}$.
    These two final reductions close the diamond on the term $t_3\esub{y}{w}$. \qedhere
  \end{itemize}
\end{proof}

\noindent
The following facts were used in the proof of \autoref{thm:diamond},
but they have merits on their own, especially
\autoref{lem:stab-struct} of stability of structures by reduction.
\begin{lem}[Stability of structures]\label{lem:stab-struct}
  Suppose $t\in\struct$ and \textup{$t \rescbnp{\rho}{\varphi}{\mu} t'$}.
  Assume that, if $\rho$ is \textsf{sub}, it applies to a variable that
  is not in $\varphi$. Then $t'\in\struct$.
\end{lem}
\begin{proof}
  By induction on $t\in\struct$, with cases on the last inference rule
  of the derivation of $t \rescbnp{\rho}{\varphi}{\mu} t'$.
\end{proof}
The restriction when $\rho$ is of kind \textsf{sub} has one simple
(but crucial) goal: ruling out a direct substitution of the leading variable
of the structure $t$,
such as $x\app u\rescbnp{\rho}{\varphi}{\mu} (\l{y}.y)\app u$
with $\rho = \rs{x}{\l{y}.y}$ and $\varphi = \set{x}$.
Indeed, any structure must be led by a \emph{frozen} variable, that is
a variable of which we know that it cannot be substituted. It can be checked
that no valid \rlsv reduction can be related to this forbidden case
of \textsf{sub}.

\begin{samepage}
\begin{lem}[Weakening]\label{lem:phi-weakening}
  Suppose \textup{$t \rescbnp{\rho}{\varphi}{\mu} t'$} and
  $\varphi\subseteq\varphi'$.
  Assume that, if $\rho$ is \textsf{sub}, it applies to a variable that
  is not in $\varphi'$. Then \textup{$t \rescbnp{\rho}{\varphi'}{\mu} t'$}.
\end{lem}
\begin{proof}
  By induction on the derivation of $t \rescbnp{\rho}{\varphi'}{\mu} t'$.
\end{proof}
\end{samepage}
\begin{lem}[Strengthening]\label{lem:phi-id-strengthening}
  If \textup{$t \rescbnp{\ri{x}}{\varphi\cup\set{x}}{\mu} t'$},
  then \textup{$t \rescbnp{\ri{x}}{\varphi}{\mu} t'$}.
\end{lem}
\begin{proof}
  By induction on the derivation of $t \rescbnp{\rho}{\varphi'}{\mu} t'$.
  There are two subtle cases with the rules \ruleAppRight and
  \ruleEsLeftFrozen, which both rely on the following property:
  If $t\in\struct[\varphi\cup\set{x}]$ and
  $t\not\in\struct[\varphi]$, then $t \rescbnp{\ri{x}}{\varphi}{\mu} t'$.
  This property is
  proved by induction on $t\in\struct[\varphi\cup\set{x}]$.
\end{proof}

\subsection{Relative optimality}

The \lscbnp-calculus is a restriction of \lscbn that requires terms to be
eagerly reduced to local normal form before they can be substituted
(\autoref{fig:aux-rules-plus}). This eager reduction is never wasted. Indeed,
\lscbn (and \emph{a fortiori} its subset \lscbnp) only reduces needed redexes, that
is, redexes that are necessarily reduced in any reduction to normal form.
As a consequence, reductions in \lscbnp are never longer than equivalent
reductions in \lscbn. On the contrary, by forcing some reductions to be
performed before a term is substituted (\emph{i.e.}, potentially duplicated), this
strategy produces in many cases reduction sequences that are strictly shorter
than the ones given by the original strong call-by-need
strategy~\cite{SCBN-ICFP}.

\begin{thm}[Minimality]\label{thm:minimal-length}
With $t'\in \nf$, if
\textup{$t \literescbn{}^n t'$} and \textup{$t \literescbnp{}^m t'$} then $m\leqslant n$.
\end{thm}

Remark that this minimality result is relative to \lscbn. The reduction
sequences of \lscbnp are not necessarily optimal with respect to the
unconstrained \lc or \l-calculi. For instance, neither \lscbnp nor
\lscbn allow reducing $r$ in the term
$(\l{x}.x\app(x\app a))\app(\l{y}.y\app r)$
prior to its duplication.

\section{Formalization in Abella}
\label{sec:abella}

We used the Coq proof assistant for our first attempts to formalize our
results. We experimented both with the locally nameless
approach~\cite{Chargueraud12} and parametric higher-order abstract
syntax~\cite{Chlipala08}. While we might eventually have succeeded using
the locally nameless approach, having to manually handle binders felt way
too cumbersome. So, we turned to a dedicated formal system, Abella~\cite{Abella14}, in
the hope that it would make syntactic proofs more straightforward. This
section describes our experience with this tool.

\subsection{Nominal variables and \texorpdfstring{$\l$-tree}{lambda-tree} syntax}

Our initial motivation for using Abella was the availability of nominal
variables through the \texttt{nabla} quantifier. Indeed, in order to open
a bound term, one has to replace the bound variable with a fresh global
variable. This freshness is critical to avoid captures; but handling it
properly causes a lot of bureaucracy in the proofs. By using nominal
variables, which are guaranteed to be fresh by the logic, this issue
disappears.

Here is an excerpt of our original definition of the \texttt{nf}
predicate, which states that a term is a normal form of \lscbn and \lscbnp,
as given in \autoref{fig:normalforms}.
The second line states that any nominal variable is in normal form, while
the third line states that an abstraction is in normal form, as long as
the abstracted term is in normal form for any nominal variable.

\begin{lstlisting}[language=Abella]
Define nf : trm -> prop by
  nabla x, nf x;
  nf (abs U) := nabla x, nf (U x);
  ...
\end{lstlisting}

Note that Abella is based on a $\lambda$-tree approach (higher-order
abstract syntax). In the above excerpt, \texttt{U} has a bound variable and
\texttt{(U x)} substitutes it with the fresh
variable~\texttt{x}. More generally, \texttt{(U V)} is the term in which
the bound variable is substituted with the term~\texttt{V}.

This approach to fresh variables was error-prone at first. Several of our
formalized theorems ended up being pointless, despite seemingly matching
the statements of our pen-and-paper formalization. Consider the following
example. This proposition states that, if \texttt{T} is a structure with
respect to \texttt{x}, and if \texttt{U} is related to \texttt{T} by
the unfolding relation \texttt{star}, then \texttt{U} is also a structure with respect
to~\texttt{x}.

\begin{lstlisting}[language=Abella]
forall T U, nabla x,
struct T x -> star T U -> struct U x.
\end{lstlisting}

Notice that the nominal variable \texttt{x} is quantified after
\texttt{T}. As a consequence, its freshness ensures that
it does not occur in \texttt{T}. Thus, the
proposition is vacuously true, since \texttt{T} cannot be
a structure with respect to a variable that does not occur in it. Had the
quantifiers been exchanged, the statement would have been fine.
Unfortunately, the design of Abella makes it much easier to use theorem statements
in which universal quantifiers happen before nominal ones, thus exacerbating the issue.
The correct way to state the above proposition is by carefully lifting
any term in which a given free variable could occur:

\begin{lstlisting}[language=Abella]
forall T U, nabla x,
struct (T x) x -> star (T x) (U x) -> struct (U x) x.
\end{lstlisting}

Once one has overcome these hurdles, advantages become apparent. For
example, to state that some free variable does not occur in
a term, not lifting this term is sufficient. And if it needed to be
lifted for some other reason, one can always equate it to a constant
$\lambda$-tree. For instance, one of our theorems needs to state that the
free variable~\texttt{x} occurring in~\texttt{T} cannot occur
in~\texttt{U}, by virtue of \texttt{star}. This is expressed by the following statement:

\begin{lstlisting}[language=Abella]
star (T x) (U x) -> exists V, U = (y\ V).
\end{lstlisting}

The $\lambda$-tree \texttt{y\textbackslash V} can be understood as the
anonymous function $y \mapsto V$. Thus, the equality above forces \texttt{U}
to be a constant $\lambda$-tree, since $y$ does not occur in $V$, by
lexical scoping.

\subsection{Functions and relations}

Our Abella formalization assumes a type \texttt{trm} and three predefined
ways to build elements of that type: application, abstraction, and
explicit substitution.

\begin{lstlisting}[language=Abella]
type app trm -> trm -> trm.
type abs (trm -> trm) -> trm.
type es  (trm -> trm) -> trm -> trm.
\end{lstlisting}

For example, a term $t \esub{x}{u}$ of our calculus will be denoted
\texttt{(es (x\textbackslash t) u)} with \texttt{t} containing some
occurrences of~\texttt{x}. Again, \texttt{x\textbackslash t} is
Abella's notation for a $\lambda$-tree \texttt{t} in which \texttt{x}
has to be substituted (by application). It can thus be understood as
an anonymous function $x \mapsto t$ in the meta-system.

Since Abella does not provide proper functions, we instead use a relation
\texttt{star} to define the unfolding function from \lc to \l. We define
it in the specification logic using \l-Prolog rules (\texttt{pi} is the
universal quantifier in the specification logic).
\begin{lstlisting}[language=Abella]
star (app U V) (app X Y) :- star U X, star V Y.
star (abs U) (abs X) :- pi x\ star x x => star (U x) (X x).
star (es U V) (X Y) :-
  star V Y, pi x\ star x x => star (U x) (X x).
\end{lstlisting}
Of particular interest is the way binders are handled; they are characterized by
stating that they are their own image: \texttt{star x x}.

Since \texttt{star} is just a relation, we have to prove that it is defined over
all the closed terms of our calculus, that it maps only to pure
\l-terms, and that it maps to exactly one \l-term. Needless
to say, all of that would be simpler if Abella had native support for
functions.

\subsection{Judgments, contexts, and derivations}

Abella provides two levels of logic: a minimal logic used for
specifications and an intuitionistic logic used for inductive reasoning
over relations. At first, we only used the reasoning logic. By doing so,
we were using Abella as if we were using Coq, except for the additional
\texttt{nabla} quantifiers. We knew of the benefits of the specification
logic when dealing with judgments and contexts; but in the case
of the untyped $\lambda$-calculus, we could not see any use for those.

Our point of view started to shift once we had to manipulate sets of free
variables, in order to distinguish which of them were frozen. We could
have easily formalized such sets by hand; but since Abella is especially
designed to handle sets of binders, we gave it a try. Let us consider the
above predicate \texttt{nf} anew, except that it is now defined using
$\lambda$-Prolog rules.

\begin{lstlisting}[language=Abella]
nf X :- frozen X.
nf (abs U) :- pi x\ frozen x => nf (U x).
nf (app U V) :- nf U, nf V, struct U.
nf (es U V) :- pi x\ frozen x => nf (U x), nf V, struct V.
nf (es U V) :- pi x\ nf (U x).
\end{lstlisting}

Specification-level propositions have the form \texttt{\{L |- P\}}, with
\texttt{P} a proposition defined in $\lambda$-Prolog and \texttt{L}
a list of propositions representing the context of \texttt{P}. Consider
the proposition \texttt{\{L |- nf (abs~T)\}}. For simplicity, let us
assume that the last three rules of \texttt{nf} do not exist (that is,
only \texttt{X} and \texttt{abs U} are covered). In that case,
there is only three ways of deriving the proposition.
Indeed, it can be derived from \texttt{\{L |- frozen (abs~T)\}} (first
rule). It can also be derived from \texttt{nabla x, \{L, frozen~x |- nf
  (T~x)\}} (second rule). Finally, the third way to derive it is if
\texttt{nf (abs~T)} is already a member of the context \texttt{L}.

The second and third derivations illustrate how Abella automates the
handling of contexts. But where Abella shines is that some theorems come
for free when manipulating specification-level properties, especially
when it comes to substitution. Suppose that one wants to prove
\texttt{\{L |- P (T~U)\}}, \emph{i.e.}, some term \texttt{T} whose bound
variable was replaced with~\texttt{U} satisfies predicate~\texttt{P} in
context \texttt{L}. The simplest way is if one can prove \texttt{nabla x,
\{L |- P (T~x)\}}. In that case, one can instantiate the nominal
variable \texttt{x} with \texttt{U} and conclude.

But more often that not, \texttt{x} occurs in the context, \emph{e.g.},
\texttt{\{L, Q~x |- P (T~x)\}} instead of \texttt{\{L |- P (T~x)\}}.
Then, proving \texttt{\{L |- P (T~U)\}} is just a matter of proving
\texttt{\{L |- Q~x\}}. But, what if the latter does not hold? Suppose
one can only prove \texttt{\{L |- R~x\}}, with \texttt{R~V :- Q~V}.
In that case, one can reason on the derivation of \texttt{\{L, Q~x |-
P (T~x)\}} and prove that \texttt{\{L, R~x |- P (T~x)\}} necessarily
holds, by definition of \texttt{R}. This ability to inductively reason
on derivations is a major strength of Abella.

Having to manipulate contexts led us to revisit most of our pen-and-paper
concepts. For example, a structure was no longer defined as a relation
with respect to its leading variable (\emph{e.g.}, \texttt{struct T x})
but with respect to all the frozen variables (\emph{e.g.},
\texttt{\{frozen~x |- struct~T\}}).
In turn, this led us to handle live variables purely through their
addition to contexts: $\varphi\cup\set{x}$. Our freshness convention is
a direct consequence, as in \autoref{fig:structures} for example.

Performing specification-level proofs does not come without its own set
of issues, though. As explained earlier, a proposition \texttt{\{L |- nf
(abs~T)\}} is derivable from the consequent being part of the context
\texttt{L}, which is fruitless. The way around it is to define
a predicate describing contexts that are well-formed, \emph{e.g.}, \texttt{L}
contains only propositions of the form \texttt{(nf~x)} with \texttt{x}
nominal. As a consequence, the case above can be eliminated because
\texttt{(abs~T)} is not a nominal variable. Unfortunately, defining these
predicates and proving the associated helper lemmas is tedious and
extremely repetitive. Thus, the user is encouraged to reuse existing
context predicates rather than creating dedicated new ones, hence leading
to sloppy and convoluted proofs. Having Abella provide some automation for
handling well-formed contexts would be a welcome improvement.

\subsection{Formal definitions}

Let us now describe the main definitions of our Abella formalization. In
addition to \texttt{nf} and \texttt{star} which have already been
presented, we have a relation \texttt{step} which formalizes the
reduction rules of \lscbn and \lscbnp presented in
\autoref{fig:scbn-rules}.

\begin{lstlisting}[language=Abella, basewidth=0.59em]
step R top (abs T) (abs T') :-
  pi x\ frozen x => step R top (T x) (T' x).
step R B (abs T) (abs T') :-
  pi x\ omega x => step R bot (T x) (T' x).
step R B (app T U) (app T' U) :- step R bot T T'.
step R B (app T U) (app T U') :- struct T, step R top U U'.
step R B (es T U) (es T' U) :-
  pi x\ omega x => step R B (T x) (T' x).
step R B (es T U) (es T' U) :-
  pi x\ frozen x => step R B (T x) (T' x), struct U.
step R B (es T U) (es T U') :-
  pi x\ active x => step (idx x) B (T x) (T x), step R bot U U'.

step (idx X) B X X :- active X.
step (sub X V) B X V :- active X.
step db B T T' :- aux_db T T'.
step lsv B T T' :- aux_lsv B T T'.
\end{lstlisting}

A small difference with respect to Section~\ref{sec:strong-cbn} is the predicate
\texttt{active}, which characterizes the variable being considered by
$\ri{x}$ (\texttt{idx}) and $\rs{x}{v}$ (\texttt{sub}). This predicate is
just a cheap way of remembering that the active variable is fresh yet not
frozen. Similarly, the predicate \texttt{omega} is used in two rules to tag
a variable as being neither frozen nor active.
Another difference is rule \ruleAbsBot. While the antecedent of the rule
is at position $\bot$, the consequent is in any position
rather than just $\bot$. Since any term reducible in position $\bot$ is
provably reducible in position $\top$, this is just a conservative
generalization of the rule.

The auxiliary rules for \lscbnp, as given in Figures~\ref{fig:aux-rules}
and~\ref{fig:aux-rules-plus} for rule \ruleAuxLsv, are as follows.

\begin{lstlisting}[language=Abella]
aux_db (app (abs T) U) (es T U).
aux_db (app (es T W) U) (es T' W) :-
  pi x\ aux_db (app (T x) U) (T' x).

aux_lsv B (es T (abs V)) (es T' (abs V)) :-
  pi x\ active x => step (sub x (abs V)) B (T x) (T' x),
  lnf bot (abs V).
aux_lsv B (es T (es U W)) (es T' W) :-
  pi x\ omega x => aux_lsv B (es T (U x)) (T' x).
aux_lsv B (es T (es U W)) (es T' W) :-
  pi x\ frozen x => aux_lsv B (es T (U x)) (T' x), struct W.
\end{lstlisting}

\noindent
Finally, an actual reduction is just comprised of rules \ruleDb and
\ruleLsv in a $\top$ position:

\begin{lstlisting}[language=Abella]
red T T' :- step db top T T'.
red T T' :- step lsv top T T'.
\end{lstlisting}

\noindent
These various rules make use of structures (\texttt{struct}), as given in
\autoref{fig:structures}.

\begin{lstlisting}[language=Abella]
struct X :- frozen X.
struct (app U V) :- struct U.
struct (es U V) :- pi x\ struct (U x).
struct (es U V) :- pi x\ frozen x => struct (U x), struct V.
\end{lstlisting}

The local norm forms of \autoref{fig:local-nf} are as follows. As for
the \texttt{step} relation, one of the rules for abstraction was
generalized with respect to Section~\ref{sec:optimal-strategies}. This time, it is for the
$\top$ position, since any term that is locally normal in a $\top$
position is locally normal in any position.

\begin{lstlisting}[language=Abella]
lnf B X :- frozen X.
lnf B X :- omega X.
lnf B (app T U) :- lnf B T, struct T, lnf top U.
lnf B (app T U) :- lnf B T, struct_omega T.
lnf B (abs T) :- pi x\ frozen x => lnf top (T x).
lnf bot (abs T) :- pi x\ omega x => lnf bot (T x).
lnf B (es T U) :- pi x\ lnf B (T x).
lnf B (es T U) :-
  pi x\ frozen x => lnf B (T x), lnf bot U, struct U.
lnf B (es T U) :- pi x\ omega x => lnf B (T x), struct_omega U.
\end{lstlisting}

Structures with respect to the set $\omega$ use a dedicated predicate
\texttt{struct\_omega}, which is just a duplicate of \texttt{struct}.
Another approach, perhaps more elegant, would have been to parameterize
\texttt{struct} with either \texttt{frozen} or \texttt{omega}.

\begin{lstlisting}[language=Abella]
struct_omega X :- omega X.
struct_omega (app U V) :- struct_omega U.
struct_omega (es U V) :- pi x\ struct_omega (U x).
struct_omega (es U V) :-
  pi x\ omega x => struct_omega (U x), struct_omega V.
\end{lstlisting}

Notice that
our Abella definition of reduction is more permissive than the restricted
calculus \lscbnp, as it allows the set $\omega$ to be non-empty when the
check \texttt{lnf bot (abs V)} is performed in the definition of
\texttt{aux\_lsv}, whereas the corresponding rule \ruleAuxLsv in the
calculus requires an empty set.

The main reason for this discrepancy is that all the variables of a
term need to appear one way or another in the context. So, we have
repurposed \texttt{omega} for variables that are neither
\texttt{active} nor \texttt{frozen}. This has avoided the introduction
of yet another predicate in the formalization. But this comes at the
expense of the set $\omega$ being potentially non-empty when
evaluating whether a term is a local normal form.

As a consequence, more terms will pass
the test and be substituted. However, this version is still included in
the \lscbn-calculus, and thus fully supports the claims of
Section~\ref{sec:soundness}.
If one were to consider a formalization of completeness instead, this
discrepancy would get in the way.

\paragraph{Extra definitions}

Having a characterization of \lscbn-terms is sometimes useful, as it
allows induction on terms rather than induction on one of the previous
predicates.

\begin{lstlisting}[language=Abella]
trm (app U V) :- trm U, trm V.
trm (abs U) :- pi x\ trm x => trm (U x).
trm (es U V) :- pi x\ trm x => trm (U x), trm V.
\end{lstlisting}

\noindent
Similarly, we might need to characterize pure \l-terms.

\begin{lstlisting}[language=Abella]
pure (app U V) :- pure U, pure V.
pure (abs U) :- pi x\ pure x => pure (U x).
\end{lstlisting}

Finally, let us remind the definitions of a $\beta$-reduction, of
a sequence \texttt{betas} of zero or more $\beta$-reductions, and of the
normal forms \texttt{nfb} of the \l-calculus, as they will be needed to
state the main theorems.

\begin{lstlisting}[language=Abella]
beta (app M N) (app M' N) :- beta M M'.
beta (app M N) (app M N') :- beta N N'.
beta (abs R) (abs R') :- pi x\ beta (R x) (R' x).
beta (app (abs R) M) (R M).

betas M M.
betas M N :- beta M P, betas P N.

nfb X :- frozen X.
nfb (abs T) :- pi x\ frozen x => nfb (T x).
nfb (app T U) :- nfb T, nfb U, notabs T.
notabs T :- frozen T.
notabs (app T U).
\end{lstlisting}

\subsection{Formally verified properties}

To conclude this section on our Abella formalization, let us state
the theorems that were fully proved using Abella.
First comes the simulation property (\autoref{lem:simulation}), which
states that, if $T\re_\rscbnp U$, then $T\unfold \re^*_\beta
U\unfold$.

\begin{lstlisting}[language=Abella]
Theorem simulation' : forall T U T* U*,
  {star T T*} -> {star U U*} -> {red T U} -> {betas T* U*}.
\end{lstlisting}

\noindent
Then comes the fact that (local) normal forms are exactly the
terms that are not reducible in \lscbnp (\autoref{lem:normal-forms-lscbn}).

\begin{lstlisting}[language=Abella]
Theorem lnf_nand_red : forall T U,
  {lnf top T} -> {red T U} -> false.
Theorem nf_nand_red : forall T U,
  {nf T} -> {red T U} -> false.

Theorem lnf_or_red : forall T,
  {trm T} -> {lnf top T} \/ exists U, {red T U}.
Theorem nf_or_red : forall T,
  {trm T} -> {nf T} \/ exists U, {red T U}.
\end{lstlisting}

\noindent
Finally, if $T$ is a normal form of \lscbn, then $T\unfold$ is a normal
form of the \l-calculus (\autoref{lem:normal-forms-lambda}).

\begin{lstlisting}[language=Abella]
Theorem nf_star' : forall T T*,
  {nf T} -> {star T T*} -> {nfb T*}.
\end{lstlisting}

\section{Abstract machine}
\label{sec:abstract}

We have also implemented an abstract machine for $\lscbnp$, so as to
get more insights about its reduction rules. Indeed, while
Figures~\ref{fig:scbn-rules} and~\ref{fig:aux-rules} already give a
presentation of the calculus that is close to be implementable, a
number of questions remain. In particular, several rules have side
conditions. So, even if the calculus is relatively optimal in the
number of reductions, it does not tell much about its practical
efficiency. For example, rules \ruleAppRight and \ruleEsLeftFrozen
require subterms to be structures, while rule \ruleEsRight involves
the reduction $\ri{x}$, which does not progress by definition. As for
rule \ruleAuxLsv of \autoref{fig:aux-rules-plus}, it requires a
subterm to be both a value and in local normal form.
Thus, an actual implementation might be nowhere near the ideal time
complexity of $O(n + m)$, with $n$ the size of the input term and $m$
the number of steps to reach its normal form.

Another practical issue that is not apparent in
Figures~\ref{fig:scbn-rules} and~\ref{fig:aux-rules} is the handling of
explicit substitutions. Indeed, the ability to freely reduce under
abstractions makes them especially brittle, as the normalization
status of their right-hand side keep switching back and forth as
variables become frozen. This phenomenon does not occur when
performing a weak reduction, even when iterating it to get the strong
normal form.

Finally, actually implementing the calculus might make apparent some
reductions in Figures~\ref{fig:scbn-rules} and~\ref{fig:aux-rules} that
needlessly break sharing between subterms.

\paragraph{Configurations of the machine}

Whether a big-step or small-step semantics is chosen to present the
abstract machine has no impact on its effectiveness. So, for
readability, we describe its big-step semantics
(\autoref{fig:abstract-machine-rules}), assuming the term being
reduced has a normal form.  Since the reduction rules do not perform
any fancy trick, it could be turned into a small-step semantics using
an explicit continuation, in a traditional fashion.

Configurations $(\Sigma,\Pi,\mu,t)$ are formed from an environment
$\Sigma$ of variables, a stack $\Pi$ of terms, a mode $\mu$, and the
currently visited term $t$. Assuming the machine terminates, it
returns a pair $(\Sigma',t')$, comprised of a new environment
$\Sigma'$ as well as a term $t'$ equivalent to~$t$, hopefully
reduced. Notice that all the rules of \autoref{fig:abstract-machine-rules}
propagate the environment $\Sigma$
along the computations and never duplicate it in any way, so it should
be understood as a global mutable store.

The argument stack $\Pi$ is used for applications; it is initially empty. When the
abstract machine needs to reduce an application $t\app u$, it pushes $u$
on the stack $\Pi$, and then focuses on reducing~$t$ (rule \ruleApp). This
stack appears on the left of the rules, but never on the right, as it is
always fully consumed by the machine. Contrarily to the variable environment
$\Sigma$ which is global, new stacks might be created on the fly to
reduce some subterms. This is not a strong requirement, though, as an
implementation could also use a global stack, by keeping track of the
current number of arguments with a counter or a stack mark.

The environment $\Sigma$ associates some data to every free variable
$x$ of the currently visited term $t$. If $x$ was originally bound by
an explicit substitution, \emph{e.g.}, $t\esub{x}{u}$, then the
environment has been extended with a new binding, which is denoted by
$\Sigma;x \mapsto u$ in rule \ruleEs.  The term associated to variable
$x$ in the environment $\Sigma$ is denoted $\Sigma(x)$, as in rule
\ruleVarAbs. This association can evolve along the computations, as
the associated term $u$ might be replaced by some reduced term
$u'$. In that case, the resulting environment is denoted by $\Sigma[x
  \mapsto u']$ as illustrated by rule \ruleVarES. This association is
eventually removed from the environment, as shown by rule \ruleEs.

If $x$, however, was originally the variable bound by an abstraction,
\emph{e.g.}, $\l{x}.t$, then the environment associates to it a
special symbol to denote whether $x$ is frozen ($\l_\top$) or not
($\l_\bot$), as illustrated by rule \ruleAbsNil. Initially, the
environment associates $\l_\top$ to every free variable of the term, if any.

Finally, the mode $\mu$ of a configuration is reminiscent of the mode
$\top$ and $\bot$ of the calculus. There is a slight
difference, though. The currently visited term is in a top-level-like
position only if $\mu$ is $\top$ and the stack $\Pi$ is empty.
Initially, the mode $\mu$ is set to $\top$.

\begin{figure}[t]
\begin{mathpar}
  \inferrule[\ruleAbsCons]{(\Sigma,\Pi,\mu,t\esub{x}{u}) \ra (\Sigma',t')}
                {(\Sigma,u\cdot\Pi,\mu,\l x.t) \ra (\Sigma',t')}
\and
  \inferrule[\ruleAbsNil]{(\Sigma;x\mapsto \l_\mu,\mathit{Nil},\mu,t) \ra (\Sigma';x\mapsto \l_\mu,t')}
                        {(\Sigma,\mathit{Nil},\mu,\l x.t) \ra (\Sigma',\l x.t')}
\\
  \inferrule[\ruleApp]{(\Sigma,u\cdot\Pi,\mu,t) \ra (\Sigma',t')}
                 {(\Sigma,\Pi,\mu,t\app u) \ra (\Sigma',t')}
                 \and
   \inferrule[\ruleEs]{(\Sigma;x\mapsto u,\Pi,\mu,t) \ra (\Sigma';x\mapsto u',t')}
                 {(\Sigma,\Pi,\mu,t\esub{x}{u}) \ra (\Sigma',t'\esub{x}{u'})}
\\
  \inferrule[\ruleVarAbsTop]{\Sigma(x) = \l_\top \and (\Sigma,\Pi,\mu,x) \ralf (\Sigma',t')}
                   {(\Sigma,\Pi,\mu,x) \ra (\Sigma',t')}
\and
  \inferrule[\ruleVarAbsBot]{\Sigma(x) = \l_\bot \and (\Sigma,\Pi,\mu,x) \ral (\Sigma',t')}
                  {(\Sigma,\Pi,\mu,x) \ra (\Sigma',t')}
\\
  \inferrule[\ruleVarAbs]{\Sigma(x) = u \and (\Sigma,\mathit{Nil},\bot,u) \ra (\Sigma',u') \and
                    u' = (\l{y}.t)\lctx \and
                   (\Sigma'[x\mapsto u'],\Pi,\mu,u') \ra (\Sigma'',u'')}
                 {(\Sigma,\Pi,\mu,x) \ra (\Sigma'',u'')}
\\
  \inferrule[\ruleVarES]{\Sigma(x) = u \hspace*{19.5pt} (\Sigma,\mathit{Nil},\bot,u) \ra (\Sigma',u') \hspace*{19.5pt}
                    \Sigma' \vdash u'\in \struct[\alpha] \hspace*{19.5pt}
                    (\Sigma'[x\mapsto u'],\Pi,\mu,x ) \ra_{\struct[\alpha]} (\Sigma'',u'')}
                  {(\Sigma,\Pi,\mu,x) \ra (\Sigma'',u'')}
\\
   \inferrule[\ruleStructPhi]{(\Sigma_0,\mathit{Nil},\top,t_1) \ra (\Sigma_1,t'_1) \and...\and
                  (\Sigma_{n-1},\mathit{Nil},\top,t_n) \ra (\Sigma_n,t'_n)}
                 {(\Sigma_0,t_1\cdot...\cdot t_n,\mu,x) \ralf (\Sigma_n,x\app t'_1\app...\app t'_n)}
\\
   \inferrule[\ruleStructOmega]{~}
                 {(\Sigma,t_1\cdot...\cdot t_n,\mu,x) \ral
                  (\Sigma,x\app t_1\app...\app t_n)}
\end{mathpar}
\caption{Big-step rules for the abstract machine.}
\label{fig:abstract-machine-rules}
\end{figure}

\begin{figure}[t]
\begin{mathpar}
  \inferrule{~}
            {\Sigma;x\mapsto \l_\top \vdash x\in \struct[\varphi]}
            \and
  \inferrule{~}
            {\Sigma;x\mapsto \l_\bot\vdash x\in \struct[\omega]}
\\
  \inferrule{\Sigma\vdash t\in \struct[\alpha]}
            {\Sigma;x\mapsto t\vdash x\in \struct[\alpha]}
            \and
  \inferrule{\Sigma\vdash t_1\in \struct[\alpha]}
            {\Sigma\vdash t_1\app t_2\in \struct[\alpha]}
            \and
  \inferrule{\Sigma;x\mapsto u\vdash t\in \struct[\alpha]}
            {\Sigma\vdash t\esub{x}{u}\in \struct[\alpha]}
\end{mathpar}
\caption{Structures, as viewed by the abstract machine.}
\label{fig:abstract-machine-struct}
\end{figure}

\paragraph{Reduction rules}

The machine implements a leftmost outermost strategy. Indeed, the
argument $u$ of an application $t\app u$ is moved to the stack while
the machine proceeds to the head $t$ (rule \ruleApp). Similarly, when
an explicit substitution $t\esub{x}{u}$ is encountered, the binding $x
\mapsto u$ is added to the environment $\Sigma$ while the focus moves
to $t$ (rule \ruleEs). Finally, when encountering an abstraction
$\l{x}.t$ while the argument stack is non-empty, the top $u$ of the stack is
popped to create an explicit substitution $t\esub{x}{u}$ (rule
\ruleAbsCons). Note that restricting the machine to a leftmost
outermost evaluation strategy is not an issue, thanks to the diamond
property.

While rules \ruleApp and \ruleAbsCons could be found in any other
abstract machine, rule \ruleEs is a bit peculiar. Indeed, once $t$
has been reduced to $t'$, the abstract machine recreates
an explicit substitution $t'\esub{x}{u'}$ with $u'$ the term associated
to $x$ (which might be $u$ or some reduced form of it). In other words,
the binding is removed from the environment and it might be added back to
it again later. This seemingly wasteful operation would not be needed in
a weak call-by-need strategy, but here it is critical. Indeed, in strong
call-by-need, the term $u'$ might refer to some abstraction variable $y$.
If this abstraction is later applied multiple times, $u'$ would further be
reduced and it would mix up the multiple values associated to $y$. Here
is an example:
\[(\l{f}. c_1\app (f\app c_2) (f\app c_3))\app (\l{y}. (c_4\app x)\esub{x}{y})\]

When reducing an abstraction $\l{x}.t$ while the stack $\Pi$ is empty,
the binding $x \mapsto \l_\mu$ is added to the environment and the
machine proceeds to reducing $t$ (rule \ruleAbsNil). This is the only
rule where the mode $\mu$ really matters, as it is used to guess
whether a variable is frozen. Once~$t$ has been reduced to $t'$, the
final value is obtained by reconstructing the
abstraction~$\l{x}.t'$. In an abstract machine using weak reduction, a
similar rule would exist, except that it would always use $x \mapsto
\l_\top$, as $x$ would always behave like a free variable.

The more complicated rules are triggered when reducing a variable $x$.
The behavior then
depends on the corresponding binding in the environment. If $x$ is
a frozen abstraction variable, then the machine creates
a structure with $x$ as the head variable (rule \ruleVarAbsTop). To do so, it
reduces all the terms $t_1,\ldots,t_n$ of the argument stack and
applies $x$ to them (auxiliary rule \ruleStructPhi).
If $x$ is an abstraction variable but not
(yet) frozen, the terms of the argument stack are still given as
arguments to $x$, but without being first reduced (rule \ruleVarAbsBot
and auxiliary rule \ruleStructOmega).

The last case corresponds to a variable $x$ bound to a term $u$. The
first step is to reduce~$u$ to $u'$ and to update the binding of $x$
(rules \ruleVarAbs and \ruleVarES).
Depending on the form of $u'$,
the evaluation will then proceed in different ways. If $u'$ is an
abstraction (possibly hidden behind layers of explicit substitutions $\lctx$),
then the evaluation proceeds with it, so as to consume the argument stack if
not empty (rule \ruleVarAbs). If $u'$ is a structure (as characterized by the
rules of \autoref{fig:abstract-machine-struct}), then two subcases are
considered, depending on whether $u'$ is a structure with respect to the
set $\varphi$ or to the set $\omega$.
In the former case, all the terms of the argument stack are reduced before being
aggregated to~$x$ (rule \ruleVarES and auxiliary rule \ruleStructPhi).
In the latter case, the terms of the argument stack are aggregated to~$x$ without
being first reduced (rule \ruleVarES and auxiliary rule \ruleStructOmega).

\paragraph{Validation}

The rules of our calculus can be applied to arbitrary subterms of the
$\l$-term being reduced. But when it comes to the abstract machine, the
reduction rules are only applied to the currently visited subterm. Thus,
even if the calculus rules are correctly implemented in the machine, if
some subterms were never visited or at the wrong time, the final term
might not be reduced as much as possible. This is illustrated by the
auxiliary rules \ruleStructPhi and \ruleStructOmega. The latter rule
can get away with not reducing the application arguments, as this term
will necessarily be visited again at a later time, once the status of
$x$ becomes clear. But for rule \ruleStructPhi, this might be the last
time the term is ever visited, so it has to be put into normal form.

The correctness of the abstract
machine would thus warrant another Abella formalization. Short of such
a formal verification, another way to increase the confidence in this
abstract machine is to validate it. We did so by applying it to all the
$\lambda$-terms of depth 5. The normal forms we obtained were compared to
the ones from a machine implementing a straightforward call-by-name
strategy. If none of the machines finished after 1500 reduction steps, we
assumed that the original term had no normal form and that both machines
agreed on it. No discrepancies were found for these millions of
$\l$-terms, hence providing us with a high confidence in the correctness
of our abstract machine.

\paragraph{Implementation details}

In our implementation, variables are represented by their De Bruijn
indices. Since there is a strict push/pop discipline for the
environment, as illustrated by rules \ruleAbsNil and \ruleEs, $\Sigma$
has been implemented using a random-access list, and the De Bruijn
indices point into it. The target cells contain either $\l_\top$,
$\l_\bot$, or the term associated to the given variable. This term is
stored in a mutable cell, so that any reduction to it are retained for
later use.

Using De Bruijn indices makes it trivial to access the environment,
but it also means that they have to be kept consistent with the
current environment. Indeed, the environment might have grown since
the time an application argument was put into the stack by rule
\ruleApp, and similarly with the bindings of the environment
itself. Thus, the stack and the environment actually contain
\emph{closures}, \emph{i.e.}, pairs of a term $u$ and the size of the
environment at the time $u$ was put into the stack or
environment. Then, a rule such as \ruleAbsCons can adjust the De
Bruijn indices of $u$ to the current environment. In a traditional
way, our implementation performs this adjustment lazily, that is, a
mark is put on $u$ and the actual adjustment is only performed once
$u$ is effectively traversed.
So, when going from $\l{x}.t$ to $t\esub{x}{u}$ in rule \ruleAbsCons,
the subterm $t$ is left completely unchanged, while the subterm $u$
is just marked as needing an adjustment of its De Bruijn indices.

Since rules \ruleVarAbs and \ruleVarES systematically force a
reduction of the term $u$ bound to~$x$ whenever $x$ is encountered,
there is a glaring inefficiency. Indeed, any evolution of the
environment between two occurrences of $x$ only concerns variables
that are meaningless to $x$. So, while the first encounter of $x$ puts
$u$ into local normal form, any subsequent encounter will traverse $u$
without modifying it any further. As a consequence, the implementation
differs from \autoref{fig:abstract-machine-rules} in that the
environment also stores the status of $u$. Initially, $u$ has status
\emph{non-evaluated}. After its first reduction, the status of $u$
changes to one of the following three \emph{evaluated} statuses: an
abstraction $(\l{y}.t)\lctx$, a structure of $\struct[\omega]$, a
structure of $\struct[\varphi]$. That way, the implementation of
\ruleVarAbs and \ruleVarES can directly skip to executing the
rightmost antecedent, depending on the status. Moreover, the status
stored in the environment is not obtained by a traversal of $u$; this
information is directly obtained from the previous evaluation. For
example, rule \ruleAbsNil necessarily returns an abstraction, while
rule \ruleEs returns whatever it obtained from the execution of its
antecedent. Thus, our implementation does not needlessly traverse
terms.

\paragraph{Observations}

First of all, the abstract machine is unfortunately not a faithful
counterpart to our calculus. The discrepancy lies in rule
\ruleVarAbs. Indeed, contrarily to rule \ruleAuxLsv of
\autoref{fig:aux-rules-plus}, the substituted term is not just a
value, it is potentially covered by explicit substitutions: $u' =
(\l{y}.t)\lctx$. As a consequence, when it comes to the explicit
substitutions of~$\lctx$, sharing between both occurrences of $u'$ in
the configuration $(\Sigma';x\mapsto u',\Pi,\mu,u')$ is lost. It is
not clear yet what the cost of recovering this sharing is.

Rules \ruleVarAbs and \ruleVarES of the abstract machine also put into
light some inefficiencies of the corresponding rules of our calculus.
Assume that the machine encounters~$x$ in mode $\top$, that the stack
$\Pi$ is empty, and that $x$ is bound to $u$ in the environment
$\Sigma$. The machine first evaluates $u$ with mode $\bot$ and updates
$\Sigma$ with the result $u'$. It then proceeds to evaluating~$u'$ in
place of $x$ with mode $\top$. But at that point, the binding of $x$
in the environment is no longer updated; it still points to $u'$, not
to $u''$. Thus, the machine will have to reduce $u'$ into~$u''$ for
other occurrences of $x$ at top-level-like positions. It would have
been more efficient to originally evaluate $u$ in mode $\top$,
contrarily to what rule \ruleEsRight of \autoref{fig:scbn-rules}
mandates.

Once this inefficiency is solved, another one becomes apparent, again
with rule \ruleVarAbs. The abstract machine evaluates $u'$ (which is
then an abstraction) in place of $x$. But there is not much point in
doing so, if the argument stack is empty. It could have
instead returned~$x$, since the binding $x \mapsto u'$ is part of the
environment, thus preserving some more sharing.

Fixing these two inefficiencies in the abstract machine means that it would implement
the reduction rules of a different calculus. This new calculus would not
be much different from the one presented in this paper, but it would no
longer have the diamond property.

\section{Conclusion}

This paper presents a $\l$-calculus dedicated to strong reduction. In the
spirit of a call-by-need strategy with explicit substitutions, it builds
on a linear substitution calculus~\cite{LSC-Standardization}.
Our calculus, however, embeds
a syntactic criterion that ensures that only needed redexes are
considered. Moreover, by delaying substitutions until they are in
so-called local normal forms rather than just values, all the reduction
sequences are of minimal length.

Properly characterizing these local normal forms proved difficult and
lots of iterations were needed until we reached the presented definition.
Our original approach relied on evaluation contexts, as in the original
presentation of a strong call-by-need strategy~\cite{SCBN-ICFP}.
While tractable, this made the proof of the diamond property long
and tedious. It is the use of Abella that led us to reconsider this
approach. Indeed, the kind of reasoning Abella favors forced us to give
up on evaluation contexts and look for reduction rules that were much
more local in nature. In turn, these changes made the relation with
typing more apparent. In hindsight, this would have avoided a large
syntactic proof in~\cite{SCBN-ICFP}. And more generally, this new
presentation of call-by-need reduction could be useful even in the
traditional weak setting.

Due to decidability, our syntactic criterion can characterize only part
of the needed redexes at a given time. All the needed reductions will
eventually happen, but detecting the neededness of a redex too late might
prevent the optimal reduction. It is an open question whether some other
simple criterion would characterize more needed redexes,
and thus potentially allow for even shorter sequences than our calculus.

Even with the current criterion, there is still work to be done. First
and foremost, the Abella formalization should be completed to at least
include the diamond property. There are also some potential improvements
to consider. For example, as made apparent by the rules of the abstract machine,
our calculus could avoid substituting
variables that are not applied (rule \ruleAuxLsv),
following~\cite{Yoshida-WeakOptimality, Accattoli-SCBV}, but it opens the
question of how to characterize the normal forms then.
Another venue for investigation is how this work interacts
with fully lazy sharing, which avoids more duplications but whose properties
are tightly related to weak reduction~\cite{FullLaziness}.

\bibliographystyle{alphaurl}
\bibliography{biblio}

\end{document}